\documentclass[aps,pra,reprint,final,floatfix,longbibliography,notitlepage]{revtex4-2}
\usepackage[utf8]{inputenc}
\usepackage{bm,eucal,amssymb,amsmath,amsthm,graphicx,color,physics,csquotes,booktabs}
\usepackage[hidelinks]{hyperref}
\usepackage[ruled,vlined]{algorithm2e}
\usepackage{enumitem,centernot}

\newtheorem{definition}{Definition}[section]
\newtheorem{theorem}{Theorem}[section]
\newtheorem{lemma}{Lemma}[section]
\newtheorem{notation}{Notation}[section]
\newtheorem{corollary}{Corollary}[section]
\newtheorem{proposition}{Proposition}[section]
\newtheorem{remark}{Remark}[section]

\DeclareMathOperator{\MS}{MS}
\DeclareMathOperator{\spec}{spec}
\DeclareMathOperator{\SU}{SU}

\graphicspath{{./figures/}}

\begin{document}

\title{
Magic-protected entanglement and Clifford-irreducible structure in magic state space
}

\author{Alejandro Borda Kuhlmann}
\author{Julián Rincón}
\affiliation{Department of Physics, Universidad de los Andes, Bogotá, D.C. 111711, Colombia}

\date{\today}

\begin{abstract}
    We develop a Clifford-orbit framework for studying magic-protected entanglement, which we refer to as \emph{magical entanglement}: the part of bipartite entanglement that remains after optimal stabilizer simplification. This construction leverages residual entanglement under Clifford reduction as a state-level organizing principle for magic state space. It defines canonical representatives, spectra, and ranks that characterize the Clifford-irreducible structure of a state. We identify two regimes of the magic-entanglement interplay. In the $T$-magic regime, local nonstabilizer resources can coexist with entanglement, but the protected component remains weak and state-dependent. In the $W$-magic regime, by contrast, entanglement is Clifford-irreducibly tied to nonstabilizerness, producing typical, strongly self-averaging behavior. Analytical examples and random-circuit numerics support a crossover from broad $T$-magic fluctuations to concentrated, Haar-like $W$-magic behavior. These results identify magical entanglement as an orbit-level diagnostic of how nonstabilizerness protects quantum correlations against Clifford reduction.
\end{abstract}

\maketitle

\tableofcontents

\section{Introduction}

Quantum computation inherits much of its power from operating beyond the stabilizer--Clifford realm, which occupies the boundary between efficient classical simulability and universal quantum computation. The stabilizer formalism provides a compact description of states and operations generated by Clifford gates, while the Gottesman-Knill theorem shows that circuits built from such gates can be simulated efficiently on a classical computer~\cite{Gottesman_1999, Aaronson_2004}. This is conceptually important: Clifford dynamics can generate highly nontrivial quantum correlations without by itself yielding quantum universality. One is therefore led to supplement stabilizer operations with non-Clifford resources, commonly described as \emph{magic} or \emph{nonstabilizerness}, to reach universality~\cite{Howard_2017, Liu_2022}.

Entanglement alone, then, is not a sufficient diagnostic of the kind of nonclassicality relevant for computational power~\cite{Jozsa_2003}. Highly entangled states can still lie entirely within the stabilizer sector, so entanglement by itself does not cleanly distinguish classically tractable from genuinely non-Clifford resources~\cite{Dowling_2025, Andreadakis_2026}. Magic provides the complementary notion, but it is also incomplete when viewed in isolation: in many-body settings, the relevant question is not only how far a state is from the stabilizer framework, but how this departure is organized in the presence of strong quantum correlations~\cite{Liu_2022, Qian_2025}. The key issue, therefore, is not entanglement alone nor magic alone, but the form in which the two appear \emph{together} in a quantum state.

A broad resource-oriented literature has accordingly developed a variety of ways to quantify magic. Early approaches emphasized quasiprobability and negativity-based characterizations within the resource theory of stabilizers~\cite{Veitch_2014}, while subsequent work introduced operational monotones such as the robustness of magic~\cite{Howard_2017}. From the perspective of simulation and synthesis, one encounters stabilizer rank~\cite{Bravyi_2019}, together with stabilizer extent, stabilizer nullity, and the dyadic monotone~\cite{Beverland_2020}. Complementary families include stabilizer R\'enyi entropies and stabilizer entropies~\cite{Leone_2022, Haug_2023, Bittel_2026}, as well as entropic and fidelity-based monotones such as the relative entropy of magic and stabilizer fidelity~\cite{Rubboli_2024}. More recently, these ideas have been extended to many-body settings and random circuit ensembles, where nonstabilizerness must be understood in the presence of extensive entanglement, spatial structure, and the non-Clifford cost of generating complex dynamics~\cite{Liu_2022, Leone_2026}. Together, these developments provide powerful ways of quantifying pure-state magic, but they mostly characterize the amount of nonstabilizerness rather than how it is organized with entanglement.

Recent work has therefore begun to move beyond global scalar quantifiers toward more structured notions of magic in many-body systems. States have been organized into coarse-grained \emph{magic classes} through circuit equivalence and convolution-group fixed points~\cite{Bu_2024}, while combined-resource measures such as nonlocal nonstabilizerness have been proposed to capture the joint presence of entanglement and magic at the state level~\cite{Qian_2025}. Related studies have shown that this interplay can induce a finer organization of state space, ranging from entanglement-dominated and magic-dominated regimes with distinct computational properties to Schmidt-orbit structure in random bipartite pure states~\cite{Gu_2025, Iannotti_2025}. Complementary progress has also emerged in operator space, where nonlocal nonstabilizerness has been linked to operator entanglement, and where local-operator entanglement has been shown to upper-bound several magic monotones while approximately coinciding with them in large random circuits~\cite{Dowling_2025, Andreadakis_2026}. At the state level, direct links between nonstabilizerness and entanglement-spectrum structure have likewise been identified, for instance through the relation between magic and entanglement-spectrum flatness~\cite{Tirrito_2024}. In parallel, notions such as long-range nonstabilizerness and pseudomagic have emphasized that the physically relevant content of magic depends on how it is distributed, protected, and accessed in many-body systems~\cite{Gu_2024, Korbany_2025}, while recent results on shallow-magic circuits show that classical hardness can depend crucially on the distribution of magic across layers rather than only on its total amount~\cite{Zhang_2025}.

Clifford-augmented tensor-network methods have shown that incorporating Clifford structure into variational ans\"atze can improve the representation of entangled states~\cite{Qian_2024_CAMPS, MasotLlima_2024}, while nonstabilizerness entanglement entropy has been proposed as an entanglement measure obtained after removing Clifford contributions~\cite{Huang_2024_NSEE}. Stabilizer disentangling provides a complementary many-body perspective, in which local Clifford operations are used to reduce entanglement and reveal how disentangling efficiency is tied to the magic content of the state~\cite{Frau_2025}. These developments are close in spirit to the viewpoint pursued here: the relevant object is not the total entanglement of a state, but the component that remains after the stabilizer-compatible contribution has been removed.

Here we develop this residual-entanglement viewpoint into a geometric state-level Clifford-orbit framework. Our starting point is the observation that the relevant question is not simply how much magic or entanglement a state possesses, but rather what part of its entanglement remains genuinely tied to non-Clifford structure after optimal stabilizer simplification. This leads naturally to the notion of \emph{magical entanglement}, or \emph{magic-protected entanglement}, designed to isolate the portion of quantum correlations that cannot be dismissed as either purely stabilizer-generated or purely local nonstabilizerness. As we show below, this perspective provides a Clifford-orbit organization of magic state space, with canonical representatives and orbit invariants, and separates states whose nonstabilizerness is effectively removable from those that retain irreducible quantum correlations of a genuinely nonstabilizer kind.

The paper is organized as follows. Section~\ref{sec:MS_math} introduces magical entanglement (MS), and the $T$- and $W$-magic regimes. The canonical form and orbit invariants are developed in Sec.~\ref{sec:MS_form}, while Sec.~\ref{sec:num_res} presents numerical evidence for the $T$- to $W$-magic crossover. Finally, Sec.~\ref{sec:discuss} discusses the Clifford-irreducible interpretation of the framework, its relation to other approaches, and future directions.

\section{Magical entanglement and regime structure}
\label{sec:MS_math}

To study the magic-entanglement interplay, we first introduce a quantity that isolates the part of entanglement that is tied to nonstabilizer structure. Rather than minimizing magic over local changes of basis, the viewpoint taken here is complementary: we minimize entanglement over the Clifford orbit of the state. The resulting value measures how much bipartite entanglement remains after all stabilizer transformations have been taken into account. This leads to the following definition.

\begin{definition}[Magical entanglement]
    Let $\ket\psi \in (\mathbb{C}^2)^{\otimes n}$ be an $n$-qubit pure state.  We define its \emph{magical entanglement} as
    \begin{equation}
        \label{eq:MS}
        \MS_{AB}( \ket\psi ) := \min_{C \in \mathcal{C}_n} S_{AB}( C \ket\psi ),
    \end{equation}
    where $\mathcal{C}_n$ is the $n$-qubit Clifford group and $S_{AB}$ is any standard bipartite entanglement measure, such as the entanglement entropy, evaluated across the bipartition $AB$. For simplicity, we write $\MS(\ket{\psi})$ when the bipartition is fixed or clear from context.
\end{definition}

By a slight abuse of notation, $S_{AB}(\ket\psi)$ denotes the value of the chosen entanglement measure for $\ket\psi$ with respect to the bipartition $AB$, computed from the corresponding reduced density matrix. Intuitively, MS quantifies the part of the entanglement that cannot be removed by Clifford operations. A nonzero value therefore signals entanglement that is protected from stabilizer reduction and tied to nonstabilizer structure.

\begin{proposition}\label{prop:cliff_orb}
    Let $\mathcal{O}_{\mathcal{C}_n}(\ket\psi)$ denote the orbit of $\ket\psi \in (\mathbb{C}^2)^{\otimes n}$ under the action of the Clifford group $\mathcal{C}_n$. Then, for any state $\ket\psi$ with a fixed bipartition $AB$:
    \begin{equation*}
        \MS_{AB}( \ket\psi ) = 0 \iff \exists \ket\phi \in \mathcal{O}_{\mathcal{C}_n}(\ket\psi): \ket\phi = \ket\phi_A \otimes \ket\phi_B
    \end{equation*}
    is separable (product state) across $AB$. The Clifford orbit is defined as
    \begin{equation}
        \label{eq:Cliff_orb}
        \mathcal{O}_{\mathcal{C}_n}(\ket\psi) := \{ C\ket\psi : C \in \mathcal{C}_n \}.
    \end{equation}
\end{proposition}

\begin{proof}
    The left-to-right implication is given by the definition of MS and the fact that zero entanglement implies separability. The right-to-left implication follows because a product state across $AB$ has zero entanglement across that bipartition.
\end{proof}

Intuitively, this proposition describes how having a zero MS value with a given (fixed) bipartition is indicative of an underlying \enquote{weakness} of the magic-entanglement relation: the state is Clifford-separable across that bipartition. Note as well that this result is a particular case of the first statement in Lemma~\ref{lemma:orbits}, when considering one of the two states to be a product.

\begin{proposition}[Clifford invariance]
    Let $\tilde{C} \in \mathcal{C}_n$; then, $\MS( \tilde{C} \ket\psi ) = \MS( \ket\psi )$. Intuitively, magical entanglement is invariant under arbitrary Clifford operations.
\end{proposition}

\begin{proof}
    We can rewrite the definition of magical entanglement in the following manner, making the equality evident:
    \begin{equation*}
    \begin{split}
        \MS( \ket\psi )
        &= \min_{C \in \mathcal{C}_n} S( C \ket\psi )
        = \min_{C' \in \mathcal{C}_n} S( C' \tilde C \ket\psi ) \\
        &= \MS( \tilde{C} \ket\psi ).
    \end{split}
    \end{equation*}
    The middle equality holds due to $\mathcal{C}_n$ being closed and thus any $C \in \mathcal{C}_n$ can be written as $C = C'\tilde{C}$ which is equivalent to $(C')^{\dagger}C = \tilde{C}$. Furthermore, note that $\mathcal{C}_n$ is a finite group and thus the minimum is achieved by some element of the group. 
\end{proof}

Since the Clifford group is finite, the minimum is always attained, and MS is constant on each Clifford orbit, serving as a label for it. This allows us to use Clifford orbits as well-defined objects. This will be a central object in this work.

\begin{proposition}\label{prop:MS_props}
Magical entanglement satisfies the following basic properties:

    \emph{(a) Positivity:} By the definition of any admissible measure of bipartite entanglement it follows that MS is positive semi-definite: $\forall \ket\psi \in (\mathbb{C}^2)^{\otimes n},\ \MS( \ket\psi ) \geqslant 0$, because $S( \ket\psi ) \geqslant 0$.

    \emph{(b) Product states:} Any product state $\ket{\phi_1} \otimes \cdots \otimes \ket{\phi_n}$, with $\ket{\phi_j} \in \mathbb C^2$, satisfies $\MS(\ket{\phi_1} \otimes \cdots \otimes \ket{\phi_n}) = 0$ for any bipartition, even when the single-qubit states are nonstabilizer.

    \emph{(c) Non-trivial states:} If $\MS(\ket\psi)>0$, then no Clifford operation maps $\ket\psi$ to a product state across the chosen bipartition. Thus, the entanglement across that cut is Clifford-irreducible.

    \emph{(d) Entanglement upper bound:} For every pure state $\ket{\psi} \in (\mathbb{C}^2)^{\otimes n}$, $\MS( \ket\psi ) \leqslant S(\ket\psi)$. This holds since the identity $I \in \mathcal C_n$ is an admissible Clifford operator in the minimization procedure. Hence, $\MS( \ket\psi ) = \min_{C \in \mathcal C_n} S(C \ket\psi) \leqslant S(I \ket\psi) = S(\ket\psi)$.
\end{proposition}

All statements follow directly from the definition of MS~\eqref{eq:MS} and standard properties of entanglement measures.

\begin{proposition}[Product-form upper bound]\label{bd:B4}
    Let $M(\ket\psi)$ be any faithful stabilizer magic monotone. Then,
    \begin{equation}
        \label{eq:ms_upper_bd}
        \MS(\ket\psi) \leqslant S(\ket\psi) \cdot \mathbf{1} \left(M(\ket\psi) > 0\right),
    \end{equation}
    where $\mathbf{1}(\cdot)$ is the indicator function over the set of $n$-qubit nonstabilizer states. In particular, if either $M(\ket\psi) = 0$ (stabilizer state), or $S(\ket\psi) = 0$ (magic product state), it follows that $\MS(\ket\psi) \equiv 0$.
\end{proposition}

\begin{proof}
    If $M(\ket\psi) = 0$, the state is stabilizer, hence Clifford-equivalent to a product state across any chosen cut; therefore
    $\MS(\ket\psi) = 0$ by definition. 
    If $M(\ket\psi) > 0$, the indicator equals $1$ and the bound reduces to $\MS(\ket\psi) \leqslant S(\ket\psi)$, which is exactly the bound proved in the entanglement upper bound.
\end{proof}

\subsection{Magical entanglement of prototypical states}

We now contrast two canonical families of states containing both entanglement and magic. These examples illustrate the motivation behind MS: in one case the entanglement is Clifford-removable, while in the other it is not. Additional examples are discussed in App.~\ref{app:typ_states}.

\paragraph*{W States.} The tripartite entangled W-state family, 
$\ket{W_n} = \frac{1}{\sqrt n}\big( \ket{100\cdots0} + \ket{010\cdots0} + \cdots + \ket{000\cdots1} \big),$ 
provides a canonical example of a state containing both entanglement and nonstabilizerness. The state is not a stabilizer state and is entangled across every nontrivial bipartition. Moreover, its Schmidt data are incompatible with a Clifford reduction to a product state across the relevant cut. Thus $\MS(\ket{W_n}) > 0$: the entanglement is Clifford-irreducible rather than Clifford-removable. This is consistent with previous work showing that $W$ states possess nonlocal magic~\cite{Odavic_2023}.

\paragraph*{Magic-infused GHZ state.} As a contrasting example, consider another nonequivalent tripartite entangled family of $n$-qubit states 
$\ket\psi = \frac{1}{\sqrt2}\left( \ket{0\cdots0} + e^{i\pi/4} \ket{1\cdots1} \right)$. 
This state can be prepared by applying a local nonstabilizer gate to one qubit and then using Clifford gates, such as CNOTs, to distribute the resulting phase coherently across the system. It follows that $\MS(\ket\psi) = 0$: applying the inverse Clifford circuit disentangles the state and leaves a nonstabilizer product state. Thus this example contains both magic and entanglement, but their relation is Clifford-removable. This contrasts with the $W$-state case, where the entanglement cannot be removed by Clifford processing.

\subsection{\texorpdfstring{$\nu$}{nu}-compressible decomposition and regime classification}

We next use the $\nu$-compressible representation theorem of Ref.~\cite{Gu_2025}. Let the stabilizer nullity of a state be $\nu(\ket\psi) := 2n - \lg(|\mathrm{StP}(\ket\psi)|),$ where $\mathrm{StP}(\ket\psi) := \mathrm{Stab}(\ket\psi) \cap \mathcal P_n $ is the stabilizer Pauli group of $\ket\psi$~\cite{Beverland_2020, Jiang_2023}. If $\ket\psi = U \ket0^{\otimes n}$ and $\nu(\ket\psi) = \nu$, then the theorem states that, for some $\nu/2 \leqslant \ell \leqslant \nu$,
\begin{equation*}
    U = C_2 \left( V_\ell \otimes I_{n - \ell} \right) C_1,
\end{equation*} 
where $C_1, C_2 \in \mathcal C_n$, $V_\ell$ is a nonstabilizer $\ell$-qubit gate, and $I_{n - \ell}$ is the identity on the remaining qubits. We refer to this factorization as the $\nu$-compressible decomposition of $\ket\psi$. It measures how complex a single nonstabilizer block must be in order to generate the nonstabilizerness of the state when the circuit is compressed into one such layer. The parameter $\ell$ encodes this complexity and is bounded by the stabilizer nullity $\nu$. For the magic-infused GHZ state, one may take $V_\ell^{\rm GHZ} = T \otimes I_{\ell - 1}$. In contrast, for $W$ states one needs a two-qubit nonstabilizer block, for example $V_\ell^W = \mathrm{CH} \otimes I_{\ell - 2}$, where $\mathrm{CH}$ is the controlled-Hadamard gate. This decomposition naturally provides two fundamentally different regimes of the magic-entanglement relation:

\begin{definition}[$T$-magic]
    A nonstabilizer state $\ket\psi$ belongs to the \emph{$T$-magic class} if its $\nu$-compressible decomposition includes exclusively local nonstabilizer gates, i.e., if $\ket\psi = U \ket0^{\otimes n}$ then,
    \begin{equation}
        \label{eq:CTC}
        U = C_2 \Biggl( \bigotimes_{j = 1}^{\ell} U_j \otimes I_{n-\ell} \Biggr) C_1,
    \end{equation}
    where the active local gates are non-Clifford, $U_j \in \mathrm{U}(2) \setminus \mathcal C_1$. We denote by $\mathcal{M}_T$ the set of $T$-magic states. The corresponding general circuit architecture is shown in Fig.~\ref{fig:CTC}.
    \begin{figure}
        \centering
        \includegraphics*[height=0.3\linewidth]{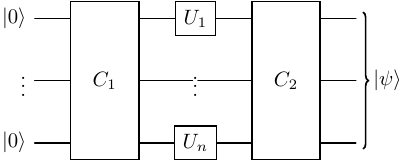}
        \caption{General CTC architecture for $T$-magic states; see Notation~\ref{not:ctc}. A layer of local nonstabilizer gates is sandwiched between Clifford circuits, defining the $1$-layer architecture family $\mathfrak{M}(1)$.}
        \label{fig:CTC}
    \end{figure}
\end{definition}

\begin{definition}[$W$-magic]
    A nonstabilizer state $\ket\psi$ belongs to the \emph{$W$-magic class} if its $\nu$-compressible decomposition includes at least a two-qubit nonstabilizer operator. We denote by $\mathcal{M}_W$ the set of $W$-magic states.
\end{definition}

Thus, within the nonstabilizer sector, $\mathcal M_T$ and $\mathcal M_W$ define complementary regimes according to whether the nonstabilizer block is local or not. Intuitively, $W$-magic represents the regime in which nonstabilizerness is tied to entanglement in a way that cannot be reproduced by a single layer of local nonstabilizer gates surrounded by Clifford operations. This does not mean that $W$-magic states cannot be generated using local nonstabilizer gates; rather, they are naturally probed by interleaving such layers with nonlocal Clifford evolution.

The most commonly used universal set of gates is the \enquote{Clifford~$+~T$} set, where $T = \text{diag}(1, e^{i\pi/4})$. More generally, it is known that $\mathcal{C}_n$ is maximal, in the sense that $\mathcal{C}_n \cup \{U\}$, with $U \in \mathrm{U}(2) \setminus \mathcal{C}_n$, yields a universal gate set~\cite{Nielsen_2011, Galindo_2026}. This motivates asking how $W$-magic can arise when the only nonstabilizer resources allowed at each layer are local gates.

\begin{notation}\label{not:ctc}
    We call a state a \enquote{CTC} state, or its matching circuit a \enquote{CTC} architecture, if it admits the decomposition in Eq.~\eqref{eq:CTC}. More generally, we consider reduced architecture families with $k$ local nonstabilizer layers interleaved with Clifford evolution. Formally, we define a nontrivial local nonstabilizer layer on $n$ qubits as
    \begin{equation}
        \breve T := \bigotimes_{j = 1}^n U_j^{p_j}, \quad 
        p_j \in \{ 0,1 \}.
    \end{equation}
    where $U_j \in \mathrm{U}(2) \setminus \mathcal C_1$ and at least one $p_j$ is nonzero. Thus, a $1$-layer architecture has the form CTC, a $2$-layer architecture has the form CTCTC, and so on.

    We use a reduced convention in which local nonstabilizer layers separated only by identity or local Clifford evolution are combined into a single layer. Hence, for $k\geq2$, intermediate Clifford stages contain nonlocal Clifford evolution. The resulting reduced architecture family is denoted by $\mathfrak M(k)$.
\end{notation}

With this convention, $\mathfrak M(0)$ is the Clifford/stabilizer family, $\mathfrak M(1)$ is the one-layer architecture associated with the $T$-magic side, and $\mathfrak M(k\geq2)$ repeatedly interleaves local nonstabilizer layers with nonlocal Clifford evolution. These reduced architecture families provide an architecture-level setting in which $W$-magic-like, Clifford-irreducible behavior can emerge, but are not intrinsic, disjoint state classes.

Although the families $\mathfrak M(k)$ are not disjoint as state sets, Sec.~\ref{sec:num_res} shows a clear statistical separation in their MS behavior. In particular, $\mathfrak M(2)$ appears as a transitional regime between the one-layer $T$-magic side and the more Haar-like behavior of $\mathfrak M(k\geqslant3)$. This makes the reduced layer structure a useful numerical probe of the $T$- to $W$-magic crossover. We now illustrate the entanglement-magic relation in the one-layer family $\mathfrak M(1)$ in terms of MS.

\begin{proposition}\label{prop:msk}
    Let $\ket\psi$ be an $n$-qubit state generated by a $1$-layer architecture, with CTC decomposition $\ket\psi = C_2 (\bigotimes_{j=1}^k V_j \otimes I_{n - k}) C_1 \ket0^{\otimes n}$, and let $S$ be the entanglement entropy. Suppose that each active local gate has the form $V_j = e^{i\theta_j P_j}$, where $P_j \in \mathcal{P}_1$, and is nontrivial and non-Clifford. Then, for any bipartition $AB$, $\MS(\ket\psi) \leqslant k$.
\end{proposition}

\begin{proof}
    The proof can be found in App.~\ref{app:msk}.
\end{proof}

The form $V_j = e^{i\theta_j P_j}$ is not as restrictive as it may seem: it includes arbitrary non-Clifford rotations about the three Pauli axes on the Bloch sphere, and in particular includes the $T$ gate. Therefore, any $1$-layer architecture using only Clifford gates and active $T$ gates satisfies the bound, with $k$ the number of active $T$ gates.

\section{MS canonical form and orbit invariants}
\label{sec:MS_form}

Since the Clifford group is finite, the minimization in the definition of MS~\eqref{eq:MS} is always attained by at least one Clifford operator. Thus $\MS_{AB}$ is a well-defined quantity for any bipartite entanglement measure $S_{AB}$ defined on pure states. We denote by $\spec_{AB}(\ket\psi)$ the ordinary Schmidt spectrum, or entanglement spectrum, of $\ket\psi$. In what follows, we fix a bipartition $AB$ of an $n$-qubit system, unless stated otherwise.

\begin{definition}[MS canonical form]
    Given a state $\ket\psi \in ( \mathbb{C}^2 )^{\otimes n}$, let $C^* \in \mathcal{C}_n$ satisfy $S_{AB}(C^* \ket\psi) = \min_{C \in \mathcal{C}_n} S_{AB}(C \ket\psi) = \MS_{AB}(\ket\psi)$, and define $\ket{\psi^*} := C^*\ket\psi$. Now let $U_A$ and $U_B$ be local unitaries acting on each partition of $AB$, such that $(U_A \otimes U_B)\ket{\psi^*} =: \ket{\psi_{\MS}}$ is in Schmidt form. Explicitly,
    \begin{equation}
        \ket{\psi_{\MS}} = \sum_{k = 1}^{r_{\MS}} \sqrt{\lambda_k} \ket k_A \otimes \ket k_B,
        \quad
        \bm{\lambda}(\ket\psi) := \left( \lambda_k \right)_{k=1}^{r_{\MS}}.
    \end{equation}
    We call $\ket{\psi_{\MS}}$ an MS canonical form of $\ket\psi$ associated with the minimizer $C^*$, $\bm{\lambda}(\ket\psi)$ the corresponding MS spectrum, and $r_{\MS}$ its MS rank. If the minimizer is not unique, the MS canonical form and the associated spectrum are understood with respect to that fixed minimizer.
\end{definition}

The scalar quantity $\MS_{AB}(\ket\psi)$ is independent of the choice of minimizer, but the spectrum of a minimizing representative need not be. To keep track of this possible nonuniqueness, one may regard the invariant fine-grained object as the set of minimizing spectra
\begin{multline}
    \Lambda_{AB}(\ket\psi) := \left\{ \spec_{AB}(C\ket\psi): C \in \mathcal C_n,\ \right. \\ 
    \left.S_{AB}(C\ket\psi) = \MS_{AB}(\ket\psi) \right\}.
\end{multline}
When this set contains a single spectrum, or when a minimizer has been fixed by convention, we write the corresponding representative spectrum as $\bm{\lambda}(\ket\psi)$.

The MS spectrum, MS rank, and MS canonical form provide a set of descriptors for analyzing the resource-theoretic structure of quantum states. These MS-derived quantities capture how nonstabilizer structure remains intertwined with entanglement after the Clifford-compatible part has been minimized. Strictly speaking, when the Clifford minimizer is not unique, the single spectrum $\bm{\lambda}(\ket\psi)$ is a chosen representative, while $\Lambda(\ket\psi)$ is the invariant spectral object associated with $\MS(\ket\psi)$. Thus, the scalar MS value gives a coarse orbit-level invariant, whereas the minimizing spectra retain finer Clifford-irreducible entanglement structure.

With this notation, we reserve $\bm{\lambda}(\ket\psi)$ for the MS spectrum of a chosen MS canonical form of $\ket\psi$. In general, $\spec(\ket\psi) \neq \bm{\lambda}(\ket\psi)$ because $\bm{\lambda}(\ket\psi)$ is computed after minimizing entanglement over the Clifford orbit. 
However, for the chosen minimizer $C^*$ we have $\spec(\ket{\psi^*}) = \bm{\lambda}(\ket\psi)$, since $\ket{\psi^*}$ only differs from its canonical form by local unitaries. Furthermore, note that by definition $\spec(\ket{\psi_{\MS}}) = \bm{\lambda}(\ket\psi)$, but $\spec(\ket{\psi_{\MS}}) \neq \bm{\lambda}(\ket{\psi_{\MS}})$ since the latter would require minimizing again over the Clifford orbit of $\ket{\psi_{\MS}}$, $\mathcal O_{\mathcal C_n}(\ket{\psi_{\MS}})$. We formalize this by resorting to two types of orbits in state space; namely, the Clifford orbit~\eqref{eq:Cliff_orb} and the (infinite) Schmidt orbit~\cite{Iannotti_2025}:
\[
    \mathcal{O}_S(\ket\psi) := \{ (U_A \otimes U_B)\ket\psi : 
    U_{A(B)} \in \mathrm{U}(2^{|A(B)|})
    \}.
\]

\begin{lemma}[Orbit invariants]\label{lemma:orbits}
    For an $n$-qubit state $\ket\psi$ take $C^*$ to be \emph{one of} the minimizers of entanglement over the Clifford orbit of $\ket\psi$, as in the previous definition. Then,
    %
    \begin{gather*}
        \ket\phi \in \mathcal{O}_{\mathcal C_n}(\ket\psi) \Longrightarrow 
        \begin{cases}
            \MS_{AB}(\ket\psi) = \MS_{AB}(\ket\phi),\\
            \Lambda_{AB}(\ket\psi) = \Lambda_{AB}(\ket\phi).
        \end{cases}
        \\
        \ket\phi \in \mathcal{O}_{S}(C^*\ket\psi) \Longrightarrow  
        \begin{cases}
            \MS_{AB}(\ket\psi) = S_{AB}(\ket\phi),\\
            \bm{\lambda}(\ket\psi) = \spec_{AB}(\ket\phi).
        \end{cases}
    \end{gather*}
\end{lemma}

\begin{proof}

    The first part follows from Clifford invariance of the minimization. If $\ket\phi = C_0 \ket\psi$, then minimizing $S_{AB}(C\ket\phi)$ over $C \in \mathcal C_n$ is equivalent to minimizing $S_{AB}(CC_0\ket\psi)$ over the same Clifford group. Hence $\MS_{AB}(\ket\phi) = \MS_{AB}(\ket\psi)$. Moreover, the Clifford minimizers of $\ket\phi$ are obtained from those of $\ket\psi$ by right multiplication with $C_0^\dagger$, so the set of minimizing spectra is unchanged: $\Lambda_{AB}(\ket\phi) = \Lambda_{AB}(\ket\psi).$

    Entanglement is invariant over the Schmidt orbit of any state. In particular, since $S_{AB}(\ket{\psi_{\MS}}) = S_{AB}(\ket{\psi^*}) = \MS_{AB}(\ket\psi)$, the same entropy and spectrum hold for any $\ket\phi \in \mathcal O_S(C^*\ket\psi)$.
\end{proof}

Note that MS is not, in general, invariant under the Schmidt orbit, even though ordinary entanglement is. This highlights a key difference to magical entanglement, which depends on a second resource (non-Cliffordness) as well as on entanglement.

\begin{proposition}[Orbit characterization via the canonical representative]\label{prop:orbit}
    Let $C_\psi$ and $C_\phi$ be chosen minimizers for $\ket\psi$ and $\ket\phi$, respectively. Then, 
    \begin{equation*}
    \begin{aligned}
        \bm{\lambda}(\ket\psi) = \bm{\lambda}(\ket\phi) \ \Longleftrightarrow \ &\exists\, U_{A(B)} \in \mathrm{U}(2^{|A(B)|}) \\ 
        &: \ket\psi = C_\psi^\dagger (U_A \otimes U_B) C_\phi \ket\phi.
    \end{aligned}
    \end{equation*}
\end{proposition}

\begin{proof}
    $(\Rightarrow)$ If the two states have identical MS spectra, by definition $\spec(\ket{\psi^*}) = \spec(\ket{\phi^*})$; thus they are LU equivalent by some local operator $U_A \otimes U_B$. Therefore,
    %
    \begin{equation*}
    \begin{aligned}
        \ket{\psi^*} &= (U_A \otimes U_B) \ket{\phi^*} \\ 
        \iff C_\psi \ket\psi &= (U_A \otimes U_B) C_\phi \ket\phi \\ 
        \iff \ket\psi &= C_\psi^\dagger (U_A \otimes U_B) C_\phi \ket\phi.
    \end{aligned}
    \end{equation*}
    $(\Leftarrow)$ Following the same steps as above allows to show that $\ket{\psi^*} \overset{\mathrm{LU}}{\displaystyle\sim} \ket{\phi^*}$ and thus they have the same spectra, which is also the MS spectrum of $|\psi\rangle$ and $\ket\phi$.
\end{proof}

This result shows that two states share the same MS spectrum only if their chosen Clifford-reduced representatives are related by a bipartition-local unitary. Equivalently, the original states are related by an operator of the form $C_\psi^\dagger (U_A \otimes U_B) C_\phi$. This is weaker than requiring a product of single-qubit unitaries, but it motivates the following specialization to $\mathfrak{M}(1)$-type transformations.

\begin{corollary} 
    For any fixed bipartition $AB$, consider two $n$-qubit states $\ket\psi$ and $\ket\phi$. If 
    \[ 
        \ket\psi 
        = 
        C_\psi^\dagger \left(U_1 \otimes \cdots \otimes U_n \right) C_\phi \ket\phi = C_\psi^\dagger \breve T C_\phi \ket\phi,
    \] 
    where the $U_j$, $(j = 1, \dots, n)$ are single-qubit unitaries and $C_\psi,\ C_\phi$ are chosen minimizers for that bipartition, then 
    \[
        \bm{\lambda}(\ket\psi) = \bm{\lambda}(\ket\phi). 
    \]
\end{corollary}

Note that, in general, we do not get the converse implication. Although equality of the MS spectra implies bipartition-local unitary equivalence of the chosen minimizing representatives for each bipartition, it does not imply that the same local operators realize the equivalence simultaneously for all cuts. Moreover, it is not even guaranteed that the transformation decomposes into a tensor product of single-qubit unitaries.

Consider a class of $T$-magic states with the simple decomposition: $\ket\psi = C ( \bigotimes ^n_{j = 1} U_j ) \ket{0^n}$, where $C \in \mathcal{C}_n$ and $U_j \in \mathrm U(2)$ for all $j$. Then, there exists $\ket{\psi'} = C' \ket\psi \in \mathcal{O}_{\mathcal{C}_n}(\ket\psi)$ (by taking $C' = C^\dagger$) such that $\ket{\psi'}$ is a product state, and hence $\MS( \ket\psi ) = 0$ and $\bm{\lambda}(\ket\psi) = (1,0,\dots,0)$. However, this does not generally hold for arbitrary $T$-magic states. For $T$-magic states with a CTC architecture, $\ket\psi = C_2 (\bigotimes_{j = 1}^n U_j ) C_1 \ket{0^n}$, and define $\ket{\psi'} = C_2^\dagger \ket\psi \in \mathcal{O}_{\mathcal{C}_n}(\ket\psi)$. The state $\ket{\psi'}$ is locally equivalent (via $\bigotimes_{j = 1}^n U_j^\dagger$) to a stabilizer state, and thus $\spec(\ket{\psi'})$ is flat. In this case, we cannot conclude that $\MS( \ket\psi ) = S(\ket{\psi'})$ or that $\bm{\lambda}(\ket\psi) = \spec(\ket{\psi'})$. For example, suppose there exists another state in the Clifford orbit of $\ket\psi$ with the same Schmidt rank as $\ket{\psi'}$. For a fixed Schmidt rank, the flat spectrum maximizes the entanglement entropy. Therefore, if another state in the same Clifford orbit has the same Schmidt rank but a less flat spectrum, $\ket{\psi'}$ cannot be the entropy minimizer.

\section{Numerical results for the MS regimes}
\label{sec:num_res}

In the previous section we introduced the $T$-magic and $W$-magic regimes and described their structural distinction in terms of MS. We now study how this distinction appears numerically in ensembles sampled from the reduced architecture families $\mathfrak{M}(k)$, with particular emphasis on the finite-size crossover from $T$-magic to $W$-magic. The main difficulty is that this behavior involves two coupled resources, entanglement and nonstabilizerness, so exact classical simulations necessarily restrict the accessible system sizes.

We choose the entanglement entropy as the measure $S_{AB}$ in the MS definition~\eqref{eq:MS}. We numerically estimate MS for states sampled from the reduced architecture families $\mathfrak{M}(1)$, $\mathfrak{M}(2)$, and $\mathfrak{M}(3)$. We implement multi-start simulated annealing~\cite{Kirkpatrick_1983} over the Clifford orbit $\mathcal{O}_{\mathcal{C}_n}(\ket\psi)$, initializing independent trajectories from random Clifford operators and minimizing the entanglement entropy. To improve exploration of the rugged landscape, we incorporate parallel tempering, allowing replicas at different temperatures to exchange configurations and escape local minima~\cite{Cicirello_2017}. The best configurations are further refined by a greedy local refinement step. The minimal value obtained across all runs is our numerical estimate of MS. More details can be found in App.~\ref{app:algo}.

Let $N$ be the sample size and $n$ the number of qubits. We calculated the MS values for $N = 100$ states sampled from $\mathfrak{M}(1)$ and $N = 50$ states sampled from $\mathfrak{M}(2)$, with $n \in [4, 12]$. For $\mathfrak{M}(3)$, we computed $N=50$ states up to $n=10$, reducing the sample size to $N = 20$ and $N = 10$ for $n = 11$ and $n = 12$, respectively, due to almost zero variance. This gives a total of $N_T = 1730$ simulated states. We use the balanced bipartition of size $\lfloor n/2\rfloor : \lceil n/2\rceil$. We do not report results for $\mathfrak{M}(k \geqslant 4)$, since the approach to Haar-random behavior is already visible for $k = 3$ at the accessible system sizes.

\begin{figure}
    \centering
    \includegraphics*[width=.8\linewidth]{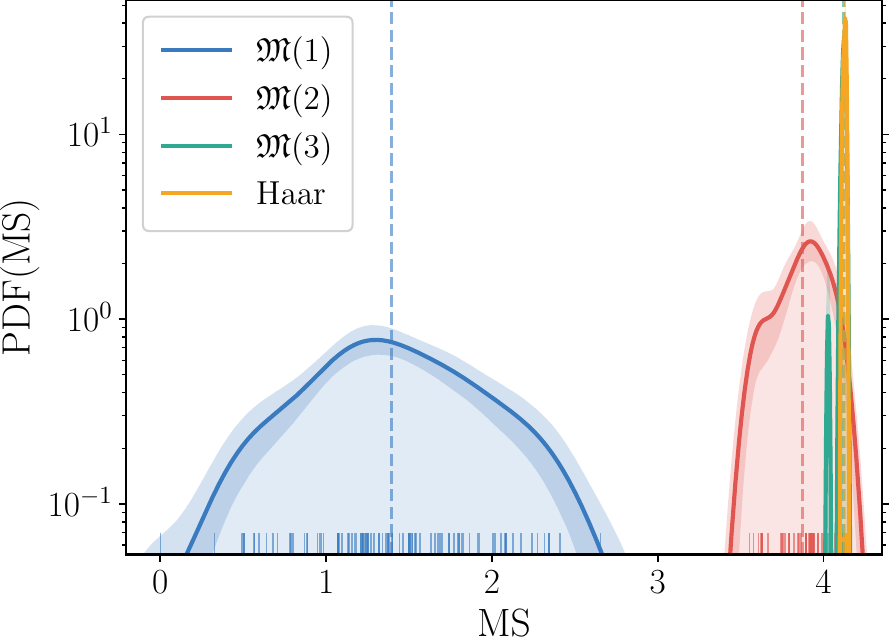}
    \caption{Kernel density estimates of the empirical MS probability distributions (PDF), and their mean (vertical dashed lines), for ensembles sampled from $\mathfrak{M}(1)$, $\mathfrak{M}(2)$, and $\mathfrak{M}(3)$, together with Haar-random states. The plot provides a qualitative visualization of the regime separation: the ensemble sampled from $\mathfrak{M}(1)$ retains a broad distribution, while the ensemble sampled from $\mathfrak{M}(3)$ becomes more sharply concentrated. The comparison with Haar-random states makes the difference in shape and variance particularly transparent.}
    \label{fig:kde_pdf}
\end{figure}

Figure~\ref{fig:kde_pdf} shows kernel density estimates of the empirical MS distributions for ensembles sampled from several $\mathfrak M(k)$ families and Haar-random states. The qualitative separation between the $T$- and $W$-magic regimes is already apparent. The $\mathfrak M(1)$ ensemble retains a broad distribution of Clifford-minimized entanglement values, indicating strongly instance-dependent protected entanglement. By contrast, the $\mathfrak M(2)$ and especially $\mathfrak M(3)$ ensembles rapidly develop sharply concentrated distributions, consistent with the onset of a typical, Haar-like magic-entanglement relation. Thus, Fig.~\ref{fig:kde_pdf} provides a qualitative visualization of the orbit-based organization of state space and its apparent crossover from $T$-magic to $W$-magic as additional reduced local nonstabilizer layers are included.

\begin{figure}
    \centering
    \includegraphics*[trim={0 45pt 0 0}, width=.8\linewidth]{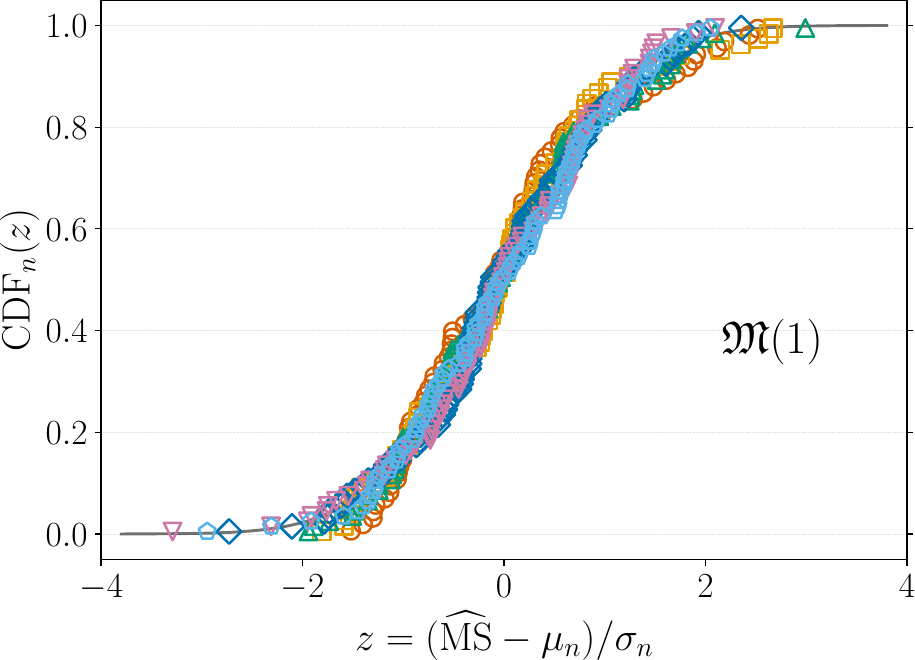}
    \includegraphics*[trim={0 45pt 0 0}, width=.8\linewidth]{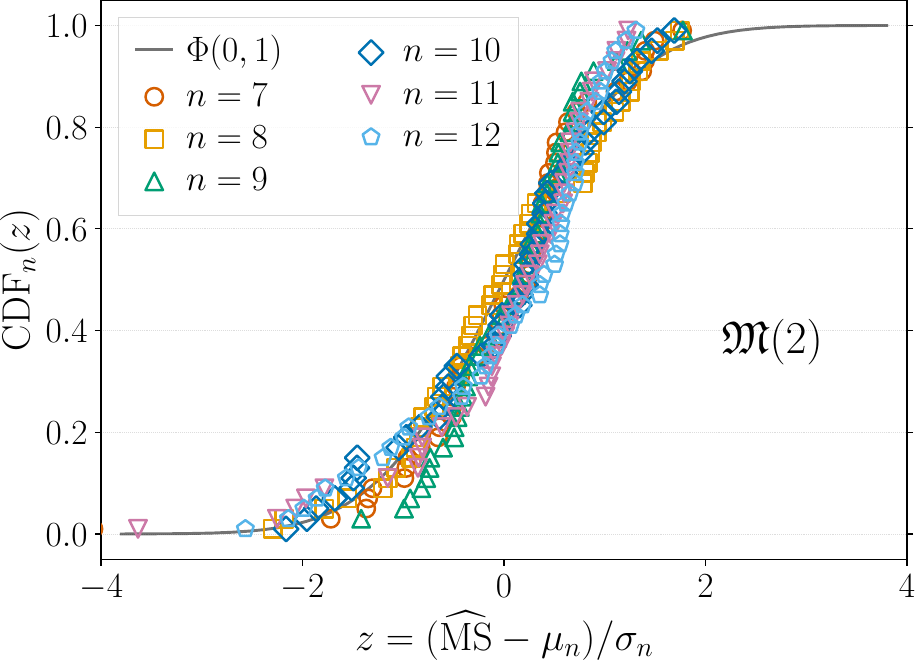}
    \includegraphics*[width=.8\linewidth]{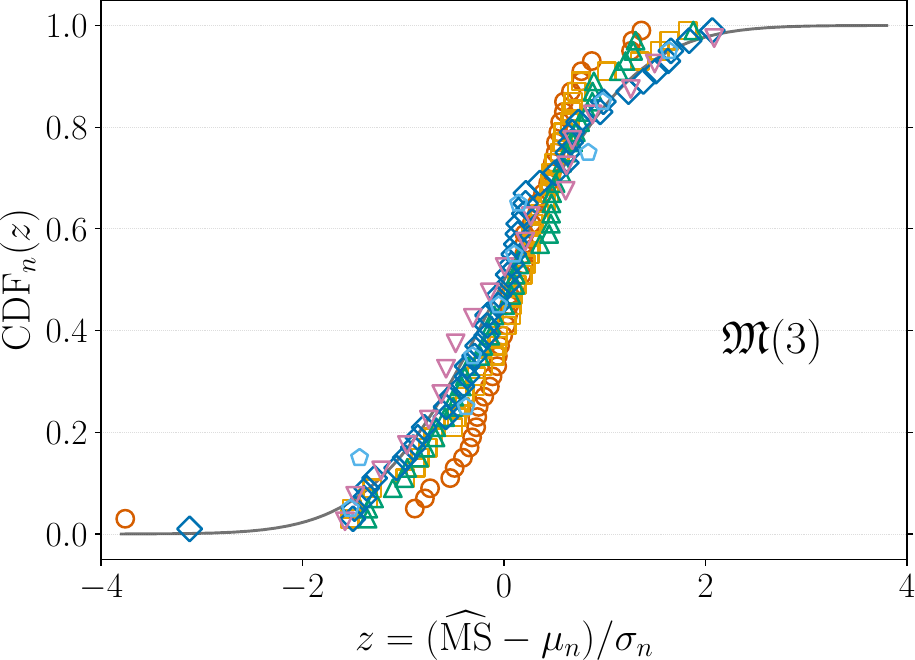}
    \caption{Empirical cumulative distribution function (CDF) collapse of the standardized positive $\widehat{\MS}$ values, $(\widehat{\MS} - \mu_n) / \sigma_n$, for ensembles sampled from $\mathfrak{M}(k)$ with $k=1,2,3$, using $n \in [7,12]$. Distinct markers indicate system size. The ensemble sampled from $\mathfrak{M}(1)$ remains comparatively broad and follows the normal reference (continuous line) more closely, consistent with approximately Gaussian finite-size fluctuations. By contrast, the ensembles sampled from $\mathfrak{M}(2)$ and $\mathfrak{M}(3)$ concentrate much more sharply, displaying a Haar-like suppression of large deviations from the typical value.}
    \label{fig:cdf-collapse}
\end{figure}

For the numerical analysis it is convenient to introduce the normalized magical entanglement
\begin{equation}
    \label{eq:MS_normd}
    \widehat{\MS}_{AB}(\ket\psi)
    :=
    \frac{\MS_{AB}(\ket\psi)}{S_{AB}(\ket\psi)},
\end{equation}
with MS and $S$ evaluated for the same bipartition $AB$. This dimensionless quantity measures the fraction of the original entanglement that is protected against Clifford reduction. By Proposition~\ref{prop:MS_props}(d), $0 \leqslant \widehat{\MS} \leqslant 1$. When $S(\ket\psi) = 0$, both the total and protected entanglement vanish, and we set $\widehat{\MS} = 0$ by convention. Thus, whereas MS measures the absolute amount of entanglement that survives optimal Clifford simplification, $\widehat{\MS}$ indicates whether it is a negligible or an order-one correction of the available bipartite correlations.

Figure~\ref{fig:cdf-collapse} shows the empirical cumulative distribution function (CDF) collapse of the standardized positive $\widehat{\MS}$ values, $(\widehat\MS - \mu_n) / \sigma_n$, for $\mathfrak{M}(k)$, with $k = 1,2,3$. The CDF is used as the main distributional tool because it avoids the binning ambiguity of histogram-based probability density function (PDF) estimates and is less sensitive to finite-sample noise. The $T$-magic ensemble $\mathfrak{M}(1)$ retains a comparatively broad fluctuation profile, whereas the $W$-magic ensembles $\mathfrak{M}(2)$ and $\mathfrak{M}(3)$ concentrate much more sharply. The standardized $\mathfrak{M}(1)$ data remain comparatively close to the normal reference, consistent with broad, approximately Gaussian finite-size fluctuations. By contrast, the $W$-magic ensembles do not simply approach a finite-width Gaussian law: their Haar-like character is reflected in the suppression of large deviations and in the tendency toward a near-delta regime. The related PDF-based diagnostics are shown in App.~\ref{app:num_res}.

\begin{figure}
    \centering
    \hspace{5pt}%
    \includegraphics*[trim={0 43pt 0 0}, width=.78\linewidth]{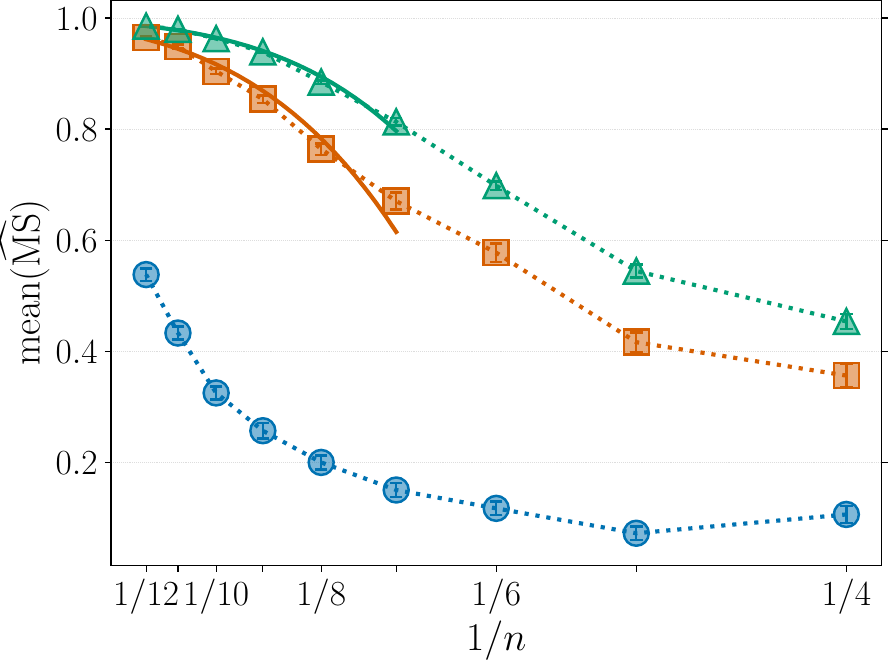}
    \includegraphics*[trim={0 43pt 0 0}, width=.8\linewidth]{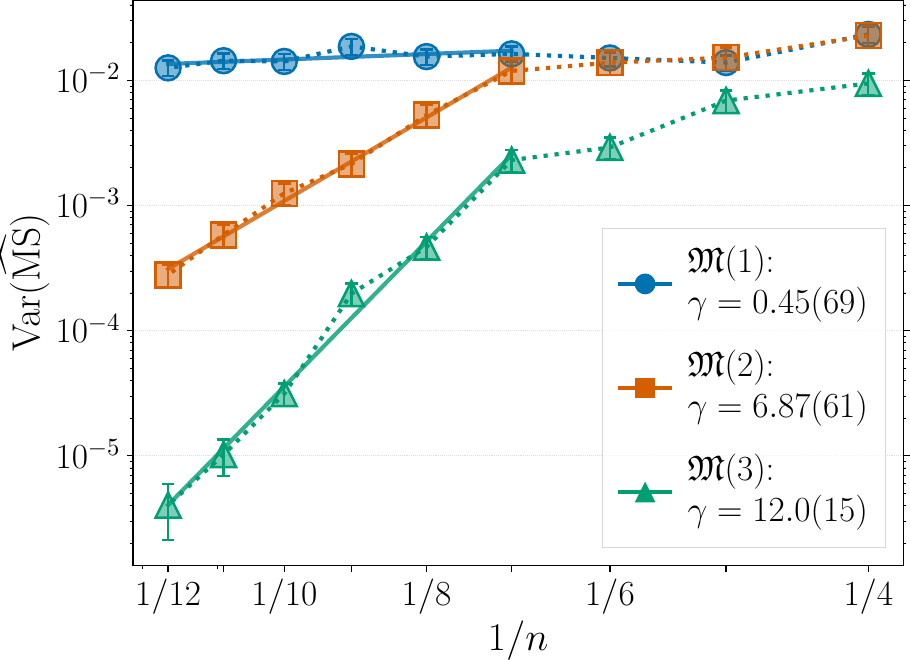}
    \includegraphics*[width=.8\linewidth]{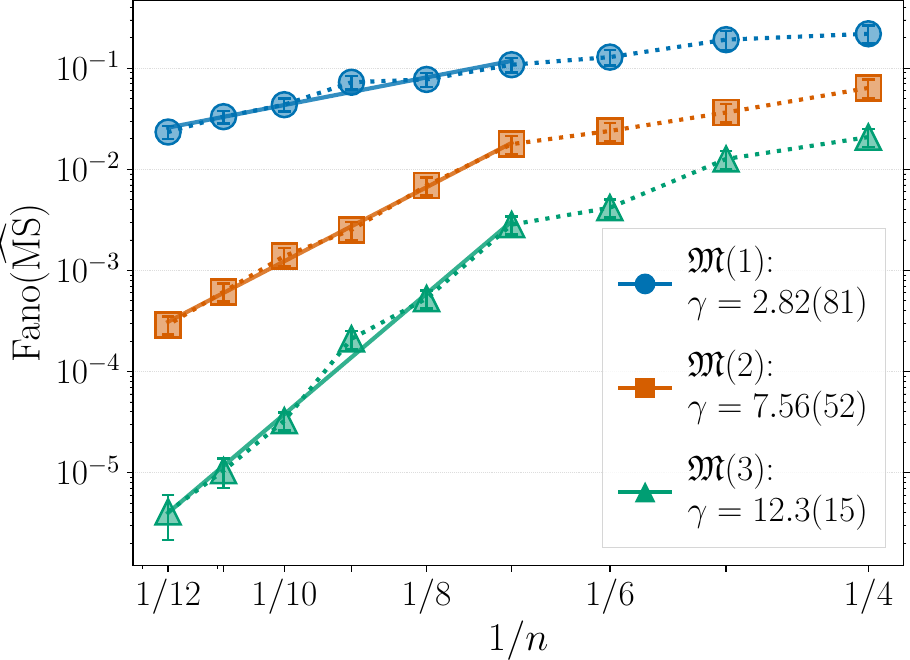}
    \caption{Moment analysis for the positive $\widehat\MS$ values. Data points for $n = 4-12$ are shown to display the early finite-size crossover, while the power-law fits (continuous lines) are performed only over $n = 7-12$. Dotted lines are guides to the eye. The mean sets the scale of $\widehat\MS$, whereas the variance and Fano factor quantify absolute and relative fluctuations. The power-law exponents are defined by $\gamma$. The rapid suppression of both fluctuation measures for ensembles sampled from $\mathfrak{M}(2)$ and $\mathfrak{M}(3)$ indicates strong self-averaging in the $W$-magic regime. The ensemble sampled from $\mathfrak{M}(1)$ remains broad and only weakly self-averaging over the same range.}
    \label{fig:moments}
\end{figure}

The moment-based analysis provides a complementary view of the finite-size behavior of $\widehat{\MS}$, separating the overall scale of the protected entanglement fraction from its absolute and relative fluctuations. As shown in Fig.~\ref{fig:moments}, data points for $n = 4, \dots, 12$ display an early finite-size crossover, while the fits are performed only over $n = 7, \dots, 12$, where the large-$n$ trends are cleaner. The mean primarily fixes the scale of $\widehat{\MS}$. Over the accessible sizes, ensembles sampled from $\mathfrak{M}(2)$ and $\mathfrak{M}(3)$ follow trends compatible with an extensive, volume-law-like behavior of typical states, while $\mathfrak{M}(1)$ is better described by a phenomenological finite-window power law with a much smaller prefactor. This suggests that $\mathfrak{M}(1)$ should not simply be viewed as a smaller-amplitude version of the higher-layer ensembles. Rather, over the sizes considered here, it behaves as an intermediate or weakly protected regime in which the Clifford-irreducible component of the entanglement has not yet developed the typical extensive character seen in the higher-layer ensembles.

The more robust statistical distinction is contained in the variance and the Fano factor, which quantify absolute and relative fluctuations. Their rapid suppression for ensembles sampled from $\mathfrak{M}(2)$ and especially $\mathfrak{M}(3)$ indicates strong self-averaging of $\widehat{\MS}$, whereas $\mathfrak{M}(1)$ remains broad and only weakly self-averaging over the accessible sizes. This behavior is consistent with the interpretation developed throughout the paper: increasing the number of nonstabilizer layers does not merely increase the amount of nonstabilizerness, but changes how this nonstabilizerness is tied to entanglement along the Clifford orbit. In the $T$-magic ensemble, the protected entanglement remains comparatively instance-dependent and does not yet define a clear volume-law-like regime; indeed, $\widehat{\MS} \ll 1$. In the $W$-magic ensembles, by contrast, $\widehat{\MS} = \Theta(1)$, indicating that an extensive fraction of the available bipartite entanglement remains protected. Thus, the higher-layer ensembles exhibit a much more typical orbit-level structure, consistent with the Haar-like organization expected of $W$-magic.

\subsection{Signatures of the \emph{T}- to \emph{W}-magic crossover}
\label{subsec:tw-crossover}

The previous analysis suggests that ensembles sampled from the reduced architecture families $\{\mathfrak M(k)\}_{k\geqslant1}$ exhibit two finite-size MS regimes, summarized in Table~\ref{tab:stats}. The $\mathfrak M(1)$ ensemble remains broad and weakly self-averaging, with $\widehat{\MS}$ well below saturation. By contrast, $\mathfrak M(2)$ and $\mathfrak M(3)$ show order-one $\widehat{\MS}$, rapidly suppressed variance and Fano factor, and distributions trending toward concentration. Thus, $\mathfrak M(k\geqslant2)$ is consistent with the onset of Dirac-like behavior for $\widehat{\MS}$, whereas $\mathfrak M(1)$ remains closer to a broad, approximately Normal finite-size fluctuation profile. This distinction reflects a change not only in the amount of nonstabilizerness, but also in how nonstabilizerness and entanglement are organized across the crossover.

\begin{table*}
\caption{
Finite-size statistical signatures of the $T$-magic and $W$-magic regimes. The table summarizes the behavior observed over the accessible system sizes. 
The last row is definitional and therefore not size-dependent.
}
\label{tab:stats}
\begin{ruledtabular}
\begin{tabular}{lll}
\textbf{MS Property}
&
\textbf{$T$-magic, $\mathfrak{M}(1)$}
&
\textbf{$W$-magic, $\mathfrak{M}(k\geqslant2)$}
\\
\midrule
Mean
&
No Haar-like saturation, $\widehat{\MS} \ll 1$
&
Close to Haar-like saturation $\widehat{\MS} \lessapprox 1$
\\
Variance
&
Weak suppression
&
Rapidly decreasing
\\
Fano factor
&
Weak self-averaging
&
Strong self-averaging
\\
Distributional behavior
&
Broad; approximately Normal
&
Concentrating toward a Dirac-like limit
\\
Physical interpretation
&
Transitional/weakly protected regime
&
Haar-like/protected-entanglement regime
\\
\midrule
Architecture
&
One reduced local layer
&
Reduced $k\geqslant2$ interleaving
\\
\end{tabular}
\end{ruledtabular}
\end{table*}

A closely related perspective comes from entanglement-spectrum statistics in random Clifford circuits with non-Clifford insertions. Ref.~\cite{Zhou_2020} showed that entanglement entropy alone does not reveal the limitations of Clifford dynamics, whereas the entanglement spectrum does: Clifford circuits yield Poisson-distributed levels, while universal non-Clifford resources lead to Wigner-Dyson statistics. This transition depends not only on inserting nonstabilizer gates, but on how they are interleaved with Clifford evolution. For nonstabilizer product states evolved by random Clifford circuits, a single inserted $T$ gate followed by further Clifford evolution drives the spectrum toward Haar-random behavior in the thermodynamic limit; stabilizer initial states do not show this transition under a single inserted $T$ layer. In our setting, this aligns with the passage from $\mathfrak M(1)$ to $\mathfrak M(2)$: the crossover reflects not just added nonstabilizerness, but a different interleaving of local nonstabilizer resources with nonlocal Clifford evolution.

Motivated by our numerical results and the MS-based $T/W$ regime structure observed in ensembles sampled from $\{ \mathfrak{M}(k) \}_{k \geqslant 1}$, the following remark summarizes the asymptotic interpretation suggested by the data.

\begin{remark}[Asymptotic Haar-like behavior]
    For every $k\geqslant2$, there exists a system size $n(k)$ such that, for $n \geqslant n(k)$, a typical $n$-qubit state sampled from $\mathfrak{M}(k)$ satisfies
    \[
        \MS(\ket{\psi}) \lessapprox S(\ket{\psi}),
        \text{ or equivalently, }
        \widehat{\MS}(\ket{\psi}) \lessapprox 1.
    \]
    Then, typically,
    \[
    \begin{aligned}
        \ket{\psi} \in \mathcal{M}_W
        & \implies
        \MS(\ket{\psi}) \lessapprox S(\ket{\psi}), 
        \\
        \ket{\psi} \in \mathcal{M}_T
        & \centernot{\implies}
        \MS(\ket{\psi}) \lessapprox S(\ket{\psi}).
    \end{aligned}
    \]
\end{remark}

The resulting physical picture is that, in the $T$-magic regime, local nonstabilizer resources can coexist with entanglement, but the Clifford-protected component remains weak and state-dependent. In the $W$-magic regime, by contrast, an order-one fraction of the bipartite entanglement survives Clifford simplification, large deviations are suppressed, and the ensemble approaches Haar-like volume-law behavior, with $\widehat{\MS}\to1$ expected in the fully Haar-like limit. The finite-$k$, finite-$n$ data therefore indicate the onset of a crossover driven not merely by increased nonstabilizerness, but by a change in how nonstabilizerness protects entanglement against Clifford reduction. This is summarized in Table~\ref{tab:stats}.

\section{Discussion and outlook}
\label{sec:discuss}

\subsection{Clifford irreducibility and MS orbit decomposition}

The main distinction uncovered in this work is not merely that different nonstabilizer states require different resources, but that they differ in how their entanglement behaves across the Clifford orbit. From this perspective, MS identifies the Clifford-irreducible component of entanglement, and the separation between $T$- and $W$-magic becomes a distinction between weakly and strongly constrained magic-entanglement correlation.

This interpretation gives a useful way of reading the numerical results of Sec.~\ref{sec:num_res}. In the $T$-magic regime, local nonstabilizer resources can coexist with entanglement, but the Clifford-irreducible component remains comparatively weak and state-dependent. In the $W$-magic regime, by contrast, this component becomes typical: large deviations are suppressed, and the ensemble approaches the Haar-like behavior expected of highly scrambled states. Thus, the numerical $T$- to $W$-magic crossover is naturally interpreted as a change in how nonstabilizerness becomes tied to entanglement.

In this language, $W$-magic states are naturally associated with a stronger form of irreducibility: their entanglement cannot be removed by even the most favorable Clifford processing. The obstruction is therefore not entanglement alone, nor nonstabilizerness alone, but the way the two resources are tied together (regardless of whether both magic and entanglement are extensive). This gives a compact formulation of the structural feature that distinguishes $W$-magic from $T$-magic.

\begin{definition}[Fundamental property of $W$-magic]\label{def:fund_prop_W}
    Let $\mathcal{E} \subset \mathcal{M}$. We say that $\mathcal{E}$ exhibits the fundamental property of \emph{$W$-magic} if, for arbitrary $\ket\psi \in \mathcal{E}$, \emph{any} Clifford gate $C \in \mathcal{C}_n$ and \emph{any} total magic-infused product state $\ket{\phi_1} \otimes \dots \otimes \ket{\phi_n}$, $C \ket\psi \neq \bigotimes_j \ket{\phi_j}$.
\end{definition}

Intuitively, this says that the entanglement of $\mathcal{E}$-like states is inherently linked to magic and lies outside the disentangling power of stabilizer operations, even when acting on magic-infused product states. The circuit depicted in Fig.~\ref{fig:c1} illustrates this idea.

\begin{figure}[b]
    \centering
    \includegraphics*[height=.3\linewidth]{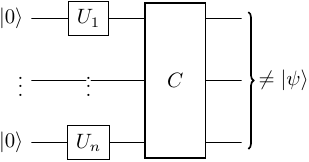}
    \caption{Circuit representation of the fundamental property of $W$-magic. A state satisfying this property has no total magic-infused product representative in its Clifford orbit. Equivalently, it cannot be generated from such a product state by Clifford processing alone.}
    \label{fig:c1}
\end{figure}

Proving such a property for \emph{all} $W$-magic states is far from simple. Nevertheless, the property holds for Haar-random states with probability $1$, as shown in App.~\ref{app:fundamental_prop}. In the same appendix we also discuss explicit families supporting this interpretation: it holds analytically for nonstabilizer hypergraph states, and we give numerical evidence for Dicke states. These examples are not meant to exhaust the $W$-magic class, but to illustrate why the irreducibility captured by MS is a natural structural feature of the $W$-magic regime.

The same idea has a geometric interpretation via the MS canonical form and the orbit invariants introduced in Sec.~\ref{sec:MS_form}. If $C^\ast$ minimizes $\MS(\ket\psi)$, then $C^\ast\ket\psi$ is a Clifford-reduced representative of the orbit. Its Schmidt form defines the MS spectrum $\bm{\lambda}(\ket\psi)$ and the MS rank $r_{\MS}$, which describe the entanglement structure left after minimizing over Clifford-compatible transformations. In this sense, the MS canonical form makes the Clifford-irreducible part of the state explicit, while $\bm{\lambda}(\ket\psi)$ and $r_{\MS}$ provide orbit-derived descriptors of that structure.

The orbit-based viewpoint also gives a geometric interpretation of the magic-entanglement interplay. Since the Clifford group acts on state space, it induces an equivalence relation whose equivalence classes are the Clifford orbits defined in Proposition~\ref{prop:cliff_orb}. These orbits are either identical or disjoint, and hence provide an orbit decomposition of state space. The definition of magical entanglement is naturally adapted to this decomposition: for a fixed bipartition, $\MS(\ket\psi)$ is the minimum entanglement attained along $\mathcal{O}_{\mathcal{C}_n}(\ket\psi)$. Thus, MS assigns to each Clifford orbit the least amount of entanglement that survives stabilizer reduction; because different orbits may share the same value, its level sets define a coarser organization of magic state space. This gives a geometric meaning to the distinction between removable and irreducible magic-entanglement structure: states in the same Clifford orbit are equivalent under stabilizer resources, while different MS values distinguish orbits with different amounts of Clifford-irreducible entanglement.

\subsection{Connections with related approaches}

A close construction to the one presented here is nonstabilizerness entanglement entropy, a residual entanglement entropy obtained after excluding Clifford-circuit contributions~\cite{Huang_2024_NSEE}. At the entropy level, this is close to MS; the difference is that here the residual-entanglement idea is developed as a Clifford-orbit framework. The definition of MS applies to general bipartite entanglement measures and leads to canonical representatives, spectral descriptors, and the $T$-/$W$-magic separation. In this sense, MS functions not only as a residual entropy, but as an organizing principle for magic state space. Clifford-augmented density matrix renormalization group and stabilizer disentangling provide complementary computational and many-body perspectives on the same separation between stabilizer-compatible and genuinely nonstabilizer entanglement~\cite{Qian_2024_CAMPS, Frau_2025}.

Related connections also arise in random-state and orbit-based studies of magic and entanglement. Stabilizer entropy, entanglement concentration, and Clifford-orbit typicality all point to regimes in which magic and entanglement acquire Haar-like statistical structure~\cite{Szombathy_2025, Bittel_2026}, while the Schmidt-orbit viewpoint highlights that nontrivial orbit structure can persist even when global statistics appear typical~\cite{Iannotti_2025}. MS complements these results by asking a different question: not how typical a Clifford orbit appears, but how much entanglement remains after the Clifford freedom has been exhausted.


The MS orbit decomposition also clarifies the relation between MS and other combined-resource measures, such as nonlocal nonstabilizerness~\cite{Qian_2025}. There, one minimizes a magic measure over local-unitary transformations, removing local-basis dependence while keeping the bipartite entanglement structure fixed; geometrically, the optimization is carried out along continuous local-unitary orbits, naturally tied to the Schmidt spectrum. The logic of MS is different but complementary: one minimizes entanglement over the Clifford orbit $\mathcal O_{\mathcal C_n}(\ket\psi)$, so the relevant equivalence classes are discrete Clifford orbits. Thus, nonlocal nonstabilizerness identifies magic robust to local-basis changes, whereas MS identifies entanglement that survives Clifford reduction. This makes MS an orbit-level invariant: it assigns to each Clifford orbit the value of MS, and its level sets provide a coarser organization of the magic state space. This distinction is natural in fault-tolerant settings, where Clifford operations are the stabilizer-preserving free operations.

\subsection{Outlook}

The central message of this work is that the relation between entanglement and nonstabilizerness is not captured by either resource separately. Stabilizer states may be highly entangled while remaining classically tractable, and nonstabilizer structure may appear without being irreducibly tied to multipartite correlations. Magical entanglement isolates the part of entanglement that remains nontrivially attached to the Clifford orbit, turning the magic-entanglement interplay into an orbit-level question. The MS canonical form, the MS spectrum, and the MS rank provide concrete descriptors of this residual structure, while the numerical results separate weakly protected $T$-magic behavior from the strongly self-averaging, Haar-like $W$-magic regime. In this sense, the formal orbit structure of Sec.~\ref{sec:MS_form} and the numerical regime structure of Sec.~\ref{sec:num_res} point to the same conclusion: MS organizes magic state space according to how nonstabilizerness protects entanglement.

A natural refinement is to study the spectral information hidden inside a fixed value of MS. The scalar MS is a coarse summary of the Clifford-minimized entanglement, while $\Lambda$ retains the Schmidt spectra of the minimizing representatives. This distinction parallels the familiar difference between entanglement entropy and the entanglement spectrum, where the latter contains information lost by a single entropy value~\cite{Li_2008}. It also suggests studying the statistics and flatness of the MS spectrum, in analogy with entanglement-spectrum diagnostics of Clifford versus non-Clifford dynamics and of nonstabilizerness~\cite{Zhou_2020, Tirrito_2024}. Majorization provides a natural way to compare such spectra, as in pure-state bipartite entanglement theory~\cite{Nielsen_1999}. In the present context, this suggests a refinement of the MS classification in which states with equal MS values may still differ in the majorization structure of their minimizing spectra.

Several directions remain open. One is how this framework extends beyond pure states, where mixed-state analogues may require optimizing over broader classes of free operations or accounting for noise in a more intrinsic way. A second direction is the multipartite problem: while bipartitions already reveal a sharp regime structure, a fuller understanding of magical entanglement in multipartite settings may require families of cuts, tensor-network formulations, or new invariants. It is also natural to ask how MS, its spectrum, and rank relate more directly to simulation complexity, scrambling, and pseudorandomness. Algorithmically, improving Clifford-orbit minimization and connecting it with Clifford-augmented and stabilizer tensor networks~\cite{Qian_2024_CAMPS, MasotLlima_2024} could make the present framework more practical for larger systems.


\begin{acknowledgments}
J.R.\ thanks M. Dalmonte, M. Collura, and C. Galindo for insightful discussions, and D. Ruge for many valuable exchanges over the years. 
J.R.\ also acknowledges support from the ICTP through the Associates Programme (2026--2031) and from the Office of the Vice President for Research and Creative Activities and the Office of the Vice President for Research of the Faculty of Science at Universidad de los Andes under the FAPA grant.
\end{acknowledgments}

\appendix

\section{Preliminaries}

The Clifford group is a central tool in quantum information theory. It is efficiently simulable on a classical computer~\cite{Gottesman_1999}, while still allowing superposition, interference, and high entanglement. More importantly, the Clifford group defines the stabilizer sector, and therefore also provides the reference structure with respect to which nonstabilizer resources are identified. The definitions discussed here can be found in greater depth in Refs.~\cite{Gottesman_1999, Aaronson_2004}.

\begin{definition}[Clifford group]
    The $n$-qubit Clifford group is defined as the normalizer of the $n$-qubit Pauli group $(\mathcal{P}_n)$ in $\SU(2^n)$, $\mathcal{C}_n := N_{\SU(2^n)}(\mathcal{P}_n)$:
    \begin{equation*}
        \mathcal{C}_n := \{ U \in \SU(2^n) : UPU^\dagger \in \mathcal{P}_n, \ \forall P\in\mathcal{P}_n \}.
    \end{equation*}
\end{definition}

\begin{definition}[Stabilizer states]
    Let $\operatorname{Stab}(\ket\psi) = \{ U \in \SU(2^n) : U\ket\psi = \ket\psi \}$, and denote the \enquote{stabilizer Pauli group} of state $\ket\psi$ by $\operatorname{StP}(\ket\psi) := \operatorname{Stab}(\ket\psi) \cap \mathcal{P}_n$. Then, the set of $n$-qubit stabilizer states is defined as
    \begin{equation*}
        S_n := \{ \ket\psi, \text{ $n$-qubit state}: |{\operatorname{StP}}(\ket\psi)| = 2^n \}.
    \end{equation*}
    Since Clifford operations preserve the size of the stabilizer Pauli group, they map stabilizer states to stabilizer states. Conversely, every $n$-qubit stabilizer state can be obtained from $\ket0^{\otimes n}$ by applying a Clifford operator.
\end{definition}

\begin{definition}[Nonstabilizer states and magic]
    Nonstabilizer states are states satisfying $|{\operatorname{StP}}(\ket\psi)| < 2^n$. Similarly, nonstabilizer operators are operators that lie outside the Clifford group. Important examples arise in higher tiers of the Clifford hierarchy~\cite{Gottesman2_1999, Cui_2017, Zeng_2008}, such as the $T$ gate in the third tier. The term \enquote{magic} is often used interchangeably with nonstabilizerness for both states and operators.
\end{definition}

\begin{definition}[Clifford hierarchy~\cite{Gottesman2_1999, Cui_2017, Zeng_2008}]
    For $n$-qubit operators, the $k$-th tier of the Clifford hierarchy is defined recursively as
    \begin{equation*}
        \mathcal{C}_n^{(k)}=\{U\in \SU(2^n): UPU^\dagger\in \mathcal{C}_n^{(k-1)},\ \forall P\in\mathcal{P}_n\},
    \end{equation*}
    where $\mathcal{C}_n^{(1)} = \mathcal{P}_n$ is the $n$-qubit Pauli group and $\mathcal{C}_n^{(2)} = \mathcal{C}_n$ is the $n$-qubit Clifford group.
\end{definition}

The first two tiers form groups, while higher tiers do not in general. Since the hierarchy is nested, $\mathcal{C}_n^{(k-1)} \subseteq \mathcal{C}_n^{(k)}$~\cite{Zeng_2008}, tiers beyond the second contain operators that are non-Clifford. In particular, the third tier, which maps Pauli strings to Clifford operations under conjugation, includes prototypical gates such as the $T$ gate, the controlled-controlled-NOT gate, also known as the Toffoli gate $CCX$, and the controlled-Hadamard gate $CH$.
Accordingly, a gate set $\mathcal G$ contains non-Clifford resources whenever $\mathcal G \cap (\mathrm U(2^n) \setminus \mathcal C_n^{(2)}) \neq \varnothing$; that is, whenever at least one gate in $\mathcal G$ lies outside the Clifford group.

\section{Further MS properties and prototypical states}
\label{app:typ_states}

\subsection{Prototypical states}

We consider here additional examples illustrating how MS distinguishes entanglement that is Clifford-removable from entanglement that remains tied to nonstabilizer structure.

\paragraph*{GHZ states.} The GHZ family represents one of the two nonequivalent classes of genuine tripartite entanglement under stochastic local operations and classical communication. It is defined by $\ket{\mathrm{GHZ}_n} = (\ket{0^n} + \ket{1^n}) / \sqrt 2$. GHZ states are stabilizer states; thus, their magic content vanishes with respect to any standard stabilizer-based magic measure. Across any $1:\mathrm{rest}$ bipartition, the entanglement entropy is $S_{\rm GHZ}(\rho_1) = \lg 2$. All of this entanglement is Clifford-removable, since $\ket{\mathrm{GHZ}_n}$ is generated from $\ket{+}^{\otimes n}$ by Clifford gates. Therefore, $\MS(\ket{\mathrm{GHZ}_n}) = 0$. This is consistent with the fact that GHZ states are Clifford-equivalent to product states (in this case, by applying $n-1$ CNOTs).

\paragraph*{CCZ state.} The CCZ magic state is defined as $\ket{\mathrm{CCZ}} = CCZ\ket{{+}{+}{+}}.$ Since CCZ is a non‑Clifford three‑qubit gate, this state lies outside the stabilizer formalism and is therefore a non‑stabilizer state. It can be viewed as a hypergraph state generated by a three‑body controlled‑phase interaction, and consequently exhibits genuine multipartite entanglement. The CCZ state is a nonstabilizer hypergraph state. As discussed in App.~\ref{app:fundamental_prop}, nonstabilizer hypergraph states satisfy the fundamental property of $W$-magic; hence no Clifford operation maps $\ket{\mathrm{CCZ}}$ to a product state. Therefore $\MS(\ket{\mathrm{CCZ}}) > 0$.

\paragraph*{Toffoli state.} A Toffoli magic state can be defined as $\ket{\mathrm{Toffoli}} = CCX \ket{{+}{+}0}$. Since $CCX = (I \otimes I \otimes H) CCZ (I \otimes I \otimes H)$, this state is local-Clifford equivalent to $\ket{\mathrm{CCZ}} = CCZ \ket{+}^{\otimes 3}$. Consequently, $\MS(\ket{\mathrm{Toffoli}}) = \MS(\ket{\mathrm{CCZ}})$, and in particular $\MS(\ket{\mathrm{Toffoli}}) > 0$.

\paragraph*{Cluster~$+~T$ states.} Consider a graph (cluster) state $\ket G$, which is a stabilizer state, and apply a $T$ gate to qubit $j$, giving $T_j\ket G$. The local $T$ gate injects nonstabilizerness into an already entangled stabilizer structure; when qubit $j$ is connected to the rest of the graph, this magic can be redistributed through the graph into nonlocal magic. The Clifford-removable part of the entanglement should therefore depend on the graph and on the chosen bipartition. For connected choices in which no Clifford operation disentangles the state across the chosen cut, this gives $\MS(T_j \ket G)>0$.

\paragraph*{Magic product state, $\ket T^{\otimes n}$.} Consider the product state $\ket T^{\otimes n}$, composed of identical single-qubit magic states $\ket T = T \ket+ = \frac{1}{\sqrt{2}}( \ket0 + e^{i \pi / 4} \ket1 )$. Being a product state, $\ket T^{\otimes n}$ contains no entanglement, although additive magic monotones scale linearly with $n$. Therefore, despite having nonzero magic, the absence of entanglement implies that $\MS(\ket T^{\otimes n}) = 0$.

\subsection{More MS properties}

\begin{notation}
    We call $U$ a \emph{TC} operator if $U = C (U_1 \otimes \dots \otimes U_n)$, for local unitaries $U_j \in \mathrm{U}(2)$ and $C \in \mathcal C_n$. Analogously, we call $U$ a \emph{CT} operator if $U = (U_1 \otimes \dots \otimes U_n) C$. A state is called a \emph{TC} state if it can be written as $U \ket0^{\otimes n}$ for some \emph{TC} operator $U$.
\end{notation}

\begin{proposition}
    The following holds:
    \begin{enumerate}[label=$\mathrm{(\roman*)}$]\itemsep0pt
        \item For any $\ket\psi \in \mathrm{TC}$, $\MS(\ket\psi) = 0$.
  
        \item For any $\ket\psi \in \mathcal{M}_T$, there exists a \emph{CT} operator $U$ such that $\MS(U \ket\psi) = 0$.
  
        \item If $\ket\psi \in \mathcal{M}_W$ satisfies Definition~\ref{def:fund_prop_W}, then $\MS(\ket\psi) > 0$.
  
        \item If $\ket\psi \in \mathcal{M}_W$ satisfies Definition~\ref{def:fund_prop_W}, then there exists no \emph{CT} operator $U$ such that $\MS(U \ket\psi) = 0$.
    \end{enumerate}
\end{proposition}

\begin{proof}
    We prove each property as follows:
    \begin{enumerate}[label=$\mathrm{(\roman*)}$]\itemsep0pt
        \item 
        Let $\ket\psi \in$~TC. Then there exist $\hat C \in \mathcal C_n$ and local unitaries $U_j$ such that $\ket\psi = \hat C( U_1 \otimes \dots \otimes U_n )\ket0^{\otimes n}$. Taking $D = \hat C^\dagger$ in the minimization gives 
        \begin{equation*}
            \MS(\ket\psi) \leqslant S(D\ket\psi) = S((U_1 \otimes \dots \otimes U_n) \ket0^{\otimes n}) = 0. 
        \end{equation*} 
        Since $\MS \geqslant 0$, it follows that $\MS(\ket\psi) = 0$.
    
    \item 
    For each $\ket\psi \in \mathcal M_T$, we have $\ket\psi = \tilde C_2 (V_1 \otimes \dots \otimes V_n) \tilde C_1 \ket0^{\otimes n}$. Choose the CT operator $U = (V_1^\dagger \otimes \dots \otimes V_n^\dagger) \tilde C_2^\dagger$. Then 
    \[
        U \ket\psi = \tilde C_1 \ket0^{\otimes n},
    \]
    which is a stabilizer state. Hence $\MS(U \ket\psi) = 0$.

    \item
    Assume $\ket\psi$ satisfies the fundamental property of $W$-magic in Definition~\ref{def:fund_prop_W}. Then no Clifford operation can disentangle $\ket\psi$ to a product state. By Proposition~\ref{prop:cliff_orb}, this is equivalent to $\MS(\ket\psi) > 0$.

    \item
    Let $U = (U_1 \otimes \dots \otimes U_n) C$ be a CT operator. Since local unitaries do not change bipartite entanglement, $S(U\ket\psi) = S(C\ket\psi)$ for every bipartition. Therefore, if $\MS(U\ket\psi) = 0$, then some Clifford operation would map $\ket\psi$ to a product state, contradicting the fundamental property of $W$-magic. \qedhere    
  \end{enumerate}
\end{proof}

\section{Proof of Proposition~\ref{prop:msk}}
\label{app:msk}

\begin{proposition}
    Let $\ket\psi$ be an $n$-qubit state generated by a $1$-layer architecture, with CTC decomposition $\ket\psi = C_2 (\bigotimes_{j=1}^k V_j \otimes I_{n - k}) C_1 \ket0^{\otimes n}$, and let $S$ be the entanglement entropy. Suppose that each active local gate has the form $V_j = e^{i\theta_j P_j}$, where $P_j \in \mathcal{P}_1$, and is nontrivial and non-Clifford. Then, for any bipartition $AB$, $\MS(\ket\psi) \leqslant k$.
\end{proposition}

\begin{proof} 
We show that there exists a Clifford operator $C$ such that $C \ket\psi$ has support on at most $2^k$ computational-basis states. Taking $C = C_1^\dagger C_2^\dagger$, we have 
\[
    C\ket\psi = C_1^\dagger \Biggl(\bigotimes_{j=1}^k V_j \otimes I_{n-k}\Biggr) C_1\ket0^{\otimes n}.
\]
Since $V_j = e^{i\theta_j P_j}$ with $P_j \in \mathcal P_1$, each factor can be written as $V_j = \cos\theta_j I + i\sin\theta_j P_j$. Therefore,
\begin{equation*}
    \bigotimes_{j=1}^k V_j = \sum_{R\subseteq [k]} c_R P_R,
\end{equation*}
where the sum runs over all subsets $R$ of $[k] = \{1,\dots,k\}$, and hence contains at most $2^k$ nonzero terms. Explicitly,
\begin{equation*}
\begin{aligned}
    c_R &= \prod_{j\in R} i\sin\theta_j \prod_{j\notin R}\cos\theta_j, 
    \\ 
    P_R &= \Biggl(\bigotimes_{j\in R} P_j\Biggr) \Biggl(\bigotimes_{j\notin R} I\Biggr) \otimes I_{n-k}.
\end{aligned}
\end{equation*}
up to the ordering of tensor factors. It follows that
\begin{equation*}
    C\ket\psi = \sum_{R\subseteq [k]} c_R C_1^\dagger P_R C_1\ket0^{\otimes n}.
\end{equation*}
Since the Clifford group normalizes the Pauli group,
\[
    Q_R := C_1^\dagger P_R C_1 \in \mathcal P_n
\]
for every $R\subseteq[k]$. Thus $Q_R\ket0^{\otimes n}$ is a computational-basis state up to an overall phase: 
$
    Q_R\ket0^{\otimes n} = e^{i\phi_R}\ket{X_R}.
$ 
Consequently,
\begin{equation*}
    C\ket\psi = \sum_{R\subseteq[k]} \tilde c_R \ket{X_R},
\end{equation*}
with at most $2^k$ nonzero computational-basis terms. Now fix an arbitrary bipartition $AB$. Since $C\ket\psi$ is supported on at most $2^k$ product basis vectors across $AB$, its Schmidt rank satisfies $\operatorname{SR}(C\ket\psi) \leqslant 2^k.$ Therefore, using the entanglement entropy,
\[
    S_{AB}(C\ket\psi) \leqslant \lg[\operatorname{SR}(C\ket\psi)] \leqslant k.
\]
Since $C\in\mathcal C_n$ is an admissible Clifford operator in the minimization defining MS, we conclude that
\[
    \MS_{AB}(\ket\psi) = \min_{D\in\mathcal C_n} S_{AB}(D\ket\psi) \leqslant S_{AB}(C\ket\psi) \leqslant k.
\]
Since the bipartition $AB$ was arbitrary, the result holds for any bipartition.
\end{proof}

\begin{algorithm*}[!t]
\SetAlgoLined 
\DontPrintSemicolon 

\caption{Multi-start Clifford-orbit minimization of entanglement with parallel tempering} \label{alg:ms-search} 

\KwIn{Initial state $\ket\psi$, parameters $(n_{\mathrm{rest}}, N_{\mathrm{steps}}, N_{\mathrm{rep}}, N_{\mathrm{swap}}, T_{\mathrm{start}}, \alpha, K, p_{\mathrm{elite}}, p_{\mathrm{cross}}, W, \gamma_{+}, \gamma_{-}, S_0)$} 

\KwOut{Best state $\ket{\phi_{\mathrm{best}}}$ and estimated minimum entropy $S_{\min}$} Initialize $S_{\min} \leftarrow \infty$, $\ket{\phi_{\mathrm{best}}} \leftarrow \ket{\psi}$, and elite pool $\mathcal{E} = \emptyset$\; Set replica temperatures $T_1 > \cdots > T_{N_{\mathrm{rep}}}$ on a ladder from $T_{\mathrm{max}} = \alpha T_{\mathrm{start}}$ to $T_{\mathrm{min}} = 10^{-4}$\; 
\;
\For{$r=1$ \KwTo $n_{\mathrm{rest}}$}{ 
    \textbf{(Initialization):} With probability $p_{\mathrm{elite}}$, seed from a state drawn uniformly from $\mathcal{E}$; otherwise seed from $\ket{\psi}$. Apply a random Clifford walk of length $n^2$ to obtain $\ket{\phi_0}$\; 
    
    \textbf{(Parallel tempering):} Run $N_{\mathrm{rep}}$ replicas for $N_{\mathrm{steps}}$ steps. At each step, propose $\ket{\phi'} = g\ket{\phi_i}$, where $g$ is drawn from the cross-cut library with probability $p_{\mathrm{cross}}$ and from the full Clifford library otherwise. Accept with probability 
    \[
        \mathbb{P}(\text{accept}) = \min\left\{ 1, e^{-(\Delta E / T_i)\ln 2}
        \right\}, 
        \qquad 
        \Delta E = E(\ket{\phi'}) - E(\ket{\phi_i}), 
    \]
    where $E$ is the sum of single-qubit entanglement entropies\; 
    
    Every $W$ steps, multiply $T_i$ by $\gamma_+$ if the recent acceptance rate is below $0.05$, by $\gamma_-$ if it is above $0.70$, and otherwise apply the baseline cooling $T_i \leftarrow \eta T_i$, with $\eta = (T_{\mathrm{min}}/T_{\mathrm{start}})^{1/N_{\mathrm{steps}}}$\; 
    
    Every $N_{\mathrm{swap}}$ steps, attempt to exchange each adjacent replica pair $(i,i+1)$ with probability 
    \[
        \min\left\{1, \exp\left[-(E_i - E_{i+1}) \left(\frac{1}{T_i} - \frac{1}{T_{i+1}}\right) \ln 2 \right] \right\}\; 
    \]
    
    \textbf{(Greedy descent):} Starting from the best state found, repeatedly apply gates from the cross-cut library and keep only steps that reduce $S_{AB}$, until no further improvement is found\;
    
    \textbf{(Elite update):} Add the resulting state to $\mathcal{E}$, retaining only the $K$ states with lowest $S_{AB}$. Update $\ket{\phi_{\mathrm{best}}}$ and $S_{\min}$ if an improvement is found\; 
    
    \If{$S_{\min} < S_0$}{
        \Return{$\ket{\phi_{\mathrm{best}}},\ 0$}\;
    }
} 
\;
\Return{$\ket{\phi_{\mathrm{best}}},\ S_{\min}$}\;
\end{algorithm*}

\section{MS search algorithm}
\label{app:algo}

We describe the numerical procedure used to estimate MS in Sec.~\ref{sec:num_res}. Since MS is defined by a Clifford-orbit minimization, the algorithm returns an upper bound on the exact value: the lowest entropy found over the explored subset of the Clifford group. The input states are randomly generated from the architecture families $\mathfrak M(k)$ by interspersing $k$ nontrivial local nonstabilizer layers with random Clifford circuits, following the reduced-layer convention of Notation~\ref{not:ctc}. In the implementation, active single-qubit gates are sampled uniformly from $\SU(2)$ and are therefore non-Clifford with probability one.

The search uses a Clifford gate library consisting of all $24$ single-qubit Clifford gates, together with CNOT, CZ, and SWAP. This library generates approximate $n$-qubit Clifford walks. Motivated by efficient canonical Clifford sampling~\cite{Selinger_2015}, each random walk has length $O(n^2)$, with gates drawn with probabilities $0.4$ for single-qubit Cliffords, $0.3$ for CNOT, $0.2$ for CZ, and $0.1$ for SWAP; acted-on qubits are chosen randomly. Given an input state $\ket\psi$, the algorithm uses a multi-restart strategy in which each restart is initialized either from $\ket\psi$ after a random Clifford walk or from an elite pool of the best states found in previous restarts. Restarts are distributed over several CPU cores.

Each restart consists of three stages. First, an initial Clifford walk of length $n^2$ is applied to the seed state, with a bias toward cross-cut gates controlled by the parameter $p_{\mathrm{cross}}$. Second, the algorithm performs a parallel-tempering search. Several replicas of the current state are evolved at different temperatures, where higher temperatures allow more frequent acceptance of moves that increase the entanglement proxy. At each step, a Clifford gate $g$ is proposed, producing $\ket{\phi'} = g \ket{\phi_i}$. The move is accepted with probability
\begin{equation*}
    \mathbb{P}(\text{accept}) = \min\left\{ 1, e^{-(\Delta E / T_i) \ln 2} \right\}, 
\end{equation*}
where $\Delta E = E(\ket{\phi'}) - E(\ket{\phi_i})$ and $E$ is the entanglement proxy. We choose $E$ to be the sum of single-qubit entanglement entropies. The true bipartite entanglement entropy $S_{AB}$ is evaluated only for promising configurations, namely after large proxy improvements or near saturation of the search. Replicas at adjacent temperatures are periodically proposed for exchange using the standard parallel-tempering acceptance rule shown in Algorithm~\ref{alg:ms-search}. Finally, the best state found during parallel tempering is refined by a greedy descent stage, in which only Clifford moves that reduce the true bipartite entropy $S_{AB}$ are accepted. The algorithm returns
\begin{equation*} 
    S_{\min} = S_{AB}(C_{\mathrm{best}} \ket\psi), 
\end{equation*}
where $C_{\mathrm{best}}$ is the Clifford walk producing the lowest entropy found during the complete search.

The algorithm contains several tunable parameters controlling exploration and computational cost, including the number of restarts, Metropolis steps, tempering replicas, swap interval, temperature schedule, elite-pool size, cross-cut gate probability, and reheating/cooling factors. Since no tight theoretical resource bound is available for this heuristic minimization, these parameters were calibrated empirically using representative system sizes. Production runs then used size-dependent parameters, with the number of replicas and Metropolis steps per restart scaled most strongly.

For reference, at $n=8$ we used $n_{\mathrm{rest}}=40$ independent restarts, $N_{\mathrm{steps}}=4000$ Metropolis steps per replica per restart, $N_{\mathrm{rep}}=5$ parallel-tempering replicas, swap interval $N_{\mathrm{swap}}=90$, maximum temperature factor $\alpha=3.20$ so that $T_{\mathrm{max}}=\alpha T_{\mathrm{start}}$, elite-pool size $K=10$, elite seeding probability $p_{\mathrm{elite}}=0.55$, cross-cut gate probability $p_{\mathrm{cross}}=0.72$, acceptance-rate window $W=250$, reheat factor $\gamma_+=1.58$, and cooling factor $\gamma_-=0.83$. The reheat and cooling factors are applied when the window acceptance rate is below $0.05$ and above $0.70$, respectively. The threshold $S_0$ is the numerical tolerance below which a state is treated as disentangled. Algorithm~\ref{alg:ms-search} summarizes the complete procedure.

\section{Additional numerical results}
\label{app:num_res}

\begin{figure}
    \centering
    \includegraphics*[trim={0 42pt 0 0}, width=.8\linewidth]{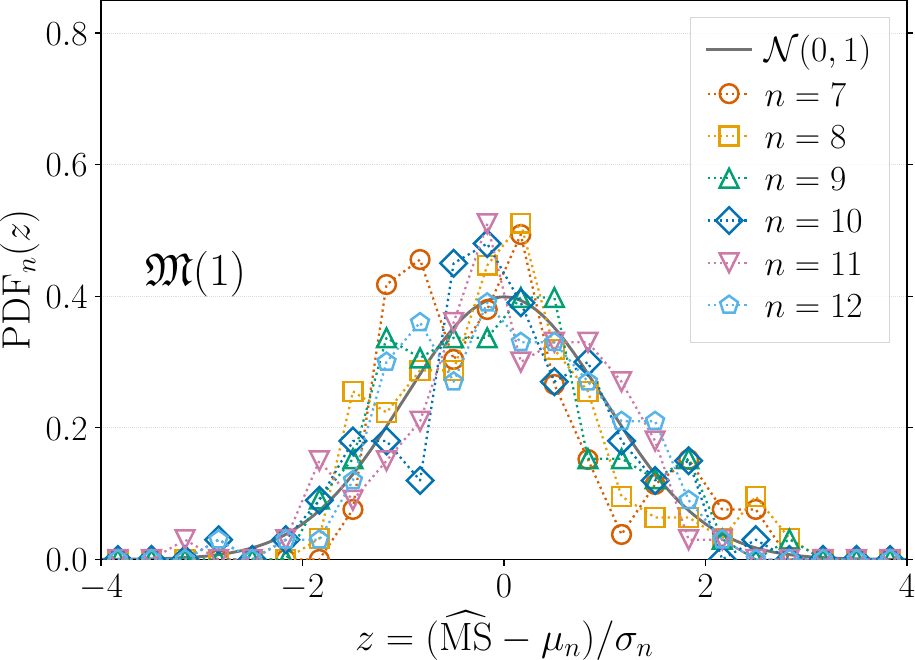}
    \includegraphics*[width=.8\linewidth]{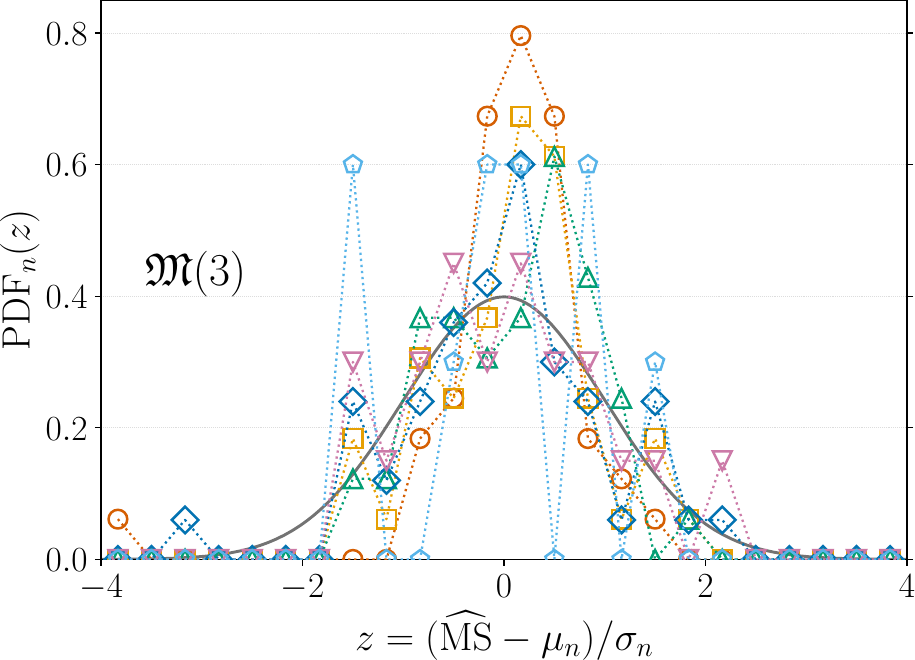}
    \caption{Histogram-based probability density function (PDF) collapse of the standardized positive $\widehat{\MS}$ values, $(\widehat{\MS} - \mu_n) / \sigma_n$, for ensembles sampled from $\mathfrak{M}(1)$ and $\mathfrak{M}(3)$, using $n \in [7,12]$. Distinct markers indicate system size. The ensemble sampled from $\mathfrak{M}(1)$ remains close to the normal reference (continuous line), whereas the ensemble sampled from $\mathfrak{M}(3)$ shows stronger concentration and clearer deviations from a finite-width Gaussian profile.}
    \label{fig:pdf-collapse}
\end{figure}

Figure~\ref{fig:pdf-collapse} shows the histogram-based probability density function (PDF) collapse of the same standardized positive $\widehat{\MS}$ values analyzed through the empirical CDF in Fig.~\ref{fig:cdf-collapse}. This plot provides a more direct visualization of the finite-size fluctuation profiles, and therefore makes the contrast between the broad $T$-magic ensemble and the more concentrated $W$-magic behavior visually transparent. However, unlike the CDF collapse used in the main text, the PDF estimate depends on the choice of binning and is more sensitive to finite-sample noise, especially for $\mathfrak{M}(3)$ at the largest system sizes, where the positive conditional sample becomes smaller. For this reason, the PDF collapse should be interpreted as a qualitative shape diagnostic rather than as the main statistical evidence. Its behavior is nevertheless consistent with the conclusions of Sec.~\ref{sec:num_res}: $\mathfrak{M}(1)$ retains a broader standardized profile, while $\mathfrak{M}(3)$ exhibits stronger concentration. Thus, Fig.~\ref{fig:pdf-collapse} serves as a robustness check supporting the CDF-based distributional analysis.

The Normal Q-Q plots in Fig.~\ref{fig:qq} provide a complementary test of the standardized fluctuation shape for the variable $(\widehat{\MS} - \mu_n) / \sigma_n$, including pointwise $95\%$ confidence bands for Normal order statistics. These plots should not be read as the main evidence for the $T$-magic/$W$-magic distinction, but rather as a shape-level check consistent with the CDF and moment analyses in Sec.~\ref{sec:num_res}. The ensemble $\mathfrak{M}(1)$ remains comparatively close to the Normal reference after standardization, whereas $\mathfrak{M}(3)$ shows systematic tail deviations. This supports the main-text interpretation: the $W$-magic distributions are not approaching a finite-width Gaussian law for $\widehat{\MS}$, but rather a strongly concentrated, near-delta regime in which standardized tail behavior becomes increasingly sensitive to finite-size and sample-size effects.

\section{More on the fundamental property of \emph{W}-magic}
\label{app:fundamental_prop}

\begin{proposition}
    With probability $1$, Haar-random states satisfy the fundamental property of $W$-magic in Definition~\ref{def:fund_prop_W}.
\end{proposition}

\begin{proof}
    Fix a bipartition and let $\mathcal{F}$ be the set of states with flat entanglement spectrum, including the rank-one product case. For fixed Schmidt rank $k$, the condition on the spectrum $\lambda_1 = \dots = \lambda_k$ imposes $k-1$ nontrivial constraints on the Schmidt spectrum. Hence, for each $k$, this set has Haar measure zero; since there are only finitely many possible Schmidt ranks, $\mathcal{F}$ has measure zero.
    
    If the entanglement spectrum of a state $\ket\eta$ can be flattened by a Clifford operator $C\in\mathcal{C}_n$, then $\ket\eta \in C^\dagger(\mathcal{F})$. Since Clifford operators are unitary, they preserve Haar measure, so each $C^\dagger(\mathcal{F})$ has measure zero. Therefore, 
    \begin{equation*}
        \bigcup_{C \in \mathcal{C}_n} C^\dagger(\mathcal{F})
    \end{equation*}
    is a finite union of measure-zero sets, and thus has zero measure.
    
    Consequently, with probability $1$, no Clifford image of a Haar-random state has flat entanglement spectrum. In particular, no Clifford image is a product state across the fixed bipartition. Since the bipartition was arbitrary, Haar-random states satisfy the fundamental property of $W$-magic with probability $1$.
\end{proof}

\begin{figure}
    \centering
    \includegraphics*[height=.52\linewidth]{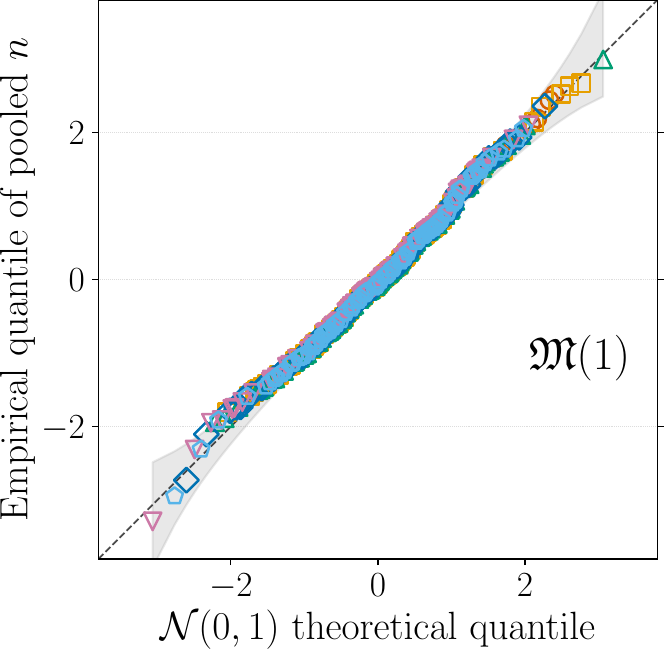}
    \includegraphics*[trim={47pt 0 0 0}, height=.52\linewidth]{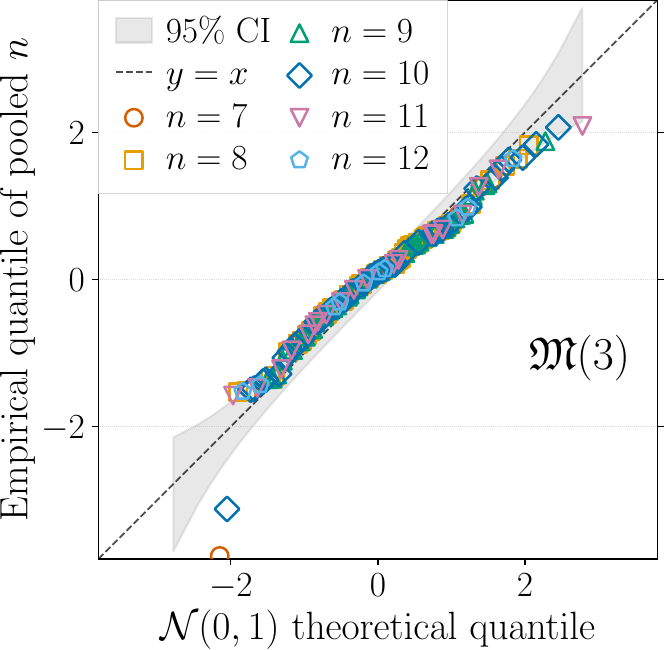}
    \caption{Normal Q-Q diagnostics for the standardized positive $\widehat{\MS}$ fluctuations, with pointwise $95\%$ confidence bands. Distinct markers indicate system size. The ensemble sampled from $\mathfrak{M}(1)$ is comparatively close to the normal reference, while that sampled from $\mathfrak{M}(3)$ shows systematic tail deviations, consistent with strong concentration rather than convergence to a finite-width Gaussian fluctuation law.}
    \label{fig:qq}
\end{figure}

\subsection{Nonstabilizer hypergraph states}

\begin{definition}[Hypergraph states~\cite{Rossi_2013, Salem_2025}]
    Let $E \subseteq \mathcal{P}(\{1, \dots, n\})$ be a set of hyperedges, where each $e \in E$ is a subset of qubits. The $n$-qubit hypergraph (HG) state associated with $E$ is defined by
    \begin{equation*}
        \ket{HG} = \prod_{e \in E} C_e \ket+^{\otimes n}.
    \end{equation*}
    Here $C_e$ denotes the multi-controlled phase gate supported on the qubits in the hyperedge $e$. Explicitly, if $e = \{e_1, \dots, e_m\}$, then $C_e$ acts as
    \begin{equation*}
        C_e = I - 2 \ket{1\cdots1}_{e} \bra{1\cdots1}_{e} \otimes I_{\bar e},
    \end{equation*}
    where $\bar e$ is the complement of $e$. Equivalently, $C_e$ applies a phase $-1$ only to computational-basis states for which all qubits in $e$ are in state $\ket1$. For $|e| = 2$, $C_e$ is the usual $CZ$ gate, while for $|e| = 3$ it is the $CCZ$ gate, and so on. The order of the qubits in $e$ is irrelevant; only the support of the hyperedge matters. The cardinality $|e|$ denotes the number of qubits involved in the hyperedge.
\end{definition}

The nonstabilizerness content of hypergraph states has been systematically characterized in the literature~\cite{Chen_2024, Salem_2025, Poderini_2026}. We will use the following criterion.

\begin{proposition}[\cite{Salem_2025, Poderini_2026}]\label{prop:salem-poder}
    Assume that the hypergraph representation of $\ket{HG}$ is reduced, i.e., repeated hyperedges have been removed. Then
    \begin{equation*}
        \ket{HG} \text{ is a stabilizer state}
        \iff
        \forall e \in E,\ |e| \leqslant 2.
    \end{equation*}
    We denote by $\mathfrak{H}$ the set of nonstabilizer hypergraph states.
\end{proposition}

\begin{proof}
    Hyperedges with $|e| = 2$ correspond to $CZ$ gates and therefore generate graph states, which are stabilizer states. Conversely, the cited characterization shows that a reduced hypergraph state containing at least one hyperedge with $|e| \geqslant 3$ is nonstabilizer. Hence $\ket{HG}$ is stabilizer if and only if all hyperedges have size at most two.
\end{proof}

\begin{lemma}\label{prop:HG}
    Let $CCZ$ denote the three-qubit controlled-controlled-$Z$ gate. For any $F \in \mathcal{C}_3$ and any local unitary $L = U_1 \otimes U_2 \otimes U_3$, with $U_j \in \SU(2)$, one has
    \begin{equation*}
        CCZ \neq F L.
    \end{equation*}
    The same obstruction applies to any multi-controlled phase gate $C_e$ with $|e| \geqslant 3$, by restricting to any three qubits in the hyperedge.
\end{lemma}

\begin{proof}
    We prove the statement for $CCZ$; the general case follows from the same three-qubit obstruction on the support of any hyperedge with $|e| \geqslant 3$. Suppose, for contradiction, that $CCZ = F L$. Then $(CCZ) L^\dagger = F \in \mathcal{C}_3$, and therefore $(CCZ) L^\dagger$ must map Pauli strings to Pauli strings under conjugation.

    \textbf{Case I:} Suppose $U_1 \in \mathcal{C}_1$. Since $U_1$ is Clifford, there exists $P_1 \in \mathcal{P}_1$ such that $U_1^\dagger P_1 U_1 = rX$ or $rY$, for some phase $r$. Taking $P = P_1 \otimes I \otimes I$, Cliffordness of $(CCZ) L^\dagger$ would imply
    \begin{equation*}
        CCZ
        (U_1^\dagger P_1 U_1 \otimes I \otimes I)
        CCZ
        \in \mathcal{P}_3.
    \end{equation*}
    If $U_1^\dagger P_1 U_1 = rX$, then
    \begin{equation*}
        r\,CCZ(X \otimes I \otimes I)CCZ = r\,X \otimes CZ_{2,3},
    \end{equation*}
    which is not a Pauli string. The case $rY$ is analogous. This contradicts $(CCZ) L^\dagger\in\mathcal{C}_3$.

    \textbf{Case II:} Suppose $U_1$ is not Clifford. Then there exists $P_1 \in \mathcal{P}_1$ such that
    \begin{equation*}
        V := U_1^\dagger P_1 U_1 \notin \mathcal{P}_1.
    \end{equation*}
    Repeating the previous argument gives
    \begin{equation*}
        CCZ(V \otimes I \otimes I)CCZ \in \mathcal{P}_3.
    \end{equation*}
    If $V \otimes I \otimes I$ commutes with $CCZ$, then the left-hand side is $V \otimes I \otimes I$, which is not a Pauli string. Otherwise, write
    \begin{equation*}
        V = c_1 I + c_2 Z + c_3 X + c_4 Y.
    \end{equation*}
    Noncommutativity with $CCZ$ implies that at least one of the $X$ or $Y$ components is nonzero. For instance, if $c_3 \neq 0$, then the conjugated operator contains the term
    \begin{equation*}
        c_3\,CCZ(X \otimes I \otimes I)CCZ = c_3\,X \otimes CZ_{2,3},
    \end{equation*}
    and therefore cannot be a Pauli string. This again contradicts $(CCZ) L^\dagger \in \mathcal{C}_3$.
\end{proof}

\begin{corollary}
    By Lemma~\ref{prop:HG}, the $\nu$-compressible decomposition of any nonstabilizer HG state cannot be written using only local nonstabilizer gates. It must contain at least one nonlocal nonstabilizer operator, and hence every nonstabilizer HG state belongs to the $W$-magic class:
    \[
        \mathfrak{H} \subset \mathcal{M}_W .
    \]
\end{corollary}

\begin{theorem}[Hypergraph states and the fundamental $W$-magic property]\label{th:HG}
    Let $\ket H$ be an $n$-qubit nonstabilizer HG state. Then, for every $C \in \mathcal{C}_n$ and every total product state $\bigotimes_{j=1}^n \ket{\phi_j}$, with $\ket{\phi_j} \in \mathbb{C}^2$, one has
    \begin{equation*}
        C\ket H \neq \bigotimes_{j=1}^n \ket{\phi_j}.
    \end{equation*}
\end{theorem}

\begin{proof}
    By definition, a HG state can be written as
    \begin{equation*}
        \ket H = \prod_{e\in E} C_e \ket+^{\otimes n}.
    \end{equation*}
    Since $\ket H$ is nonstabilizer, Proposition~\ref{prop:salem-poder} implies that at least one hyperedge $\hat e \in E$ satisfies $|\hat e| \geqslant 3$. It is enough to isolate this hyperedge, since the remaining hyperedge gates can be absorbed into the surrounding Clifford and local factors for the purpose of testing whether a total product representative exists. Thus the obstruction already appears in the minimal case $E = \{\hat e\}$ with $|\hat e| = 3$, where $C_{\hat e} = CCZ_{\hat e}$.

    Suppose, by contradiction, that the result does not hold. Then, for some $C \in \mathcal C_n$ and local unitaries $U_1, \dots, U_n$, one has
    \begin{equation*}
        \ket H = C (U_1 \otimes \dots \otimes U_n) \ket0^{\otimes n}.
    \end{equation*}
    In the minimal case, we also have
    \begin{equation*}
        \ket H = CCZ_{\hat e} H^{\otimes n} \ket0^{\otimes n}.
    \end{equation*}
    Comparing the two decompositions and absorbing the local Hadamards into the local unitaries gives, for local unitaries $V_i$, a decomposition of the form
    \begin{equation*}
        CCZ_{\hat e} = C (V_1 \otimes \dots \otimes V_n).
    \end{equation*}
    This contradicts Lemma~\ref{prop:HG}. Hence no Clifford image of $\ket H$ can be a total product state.
\end{proof}

This theorem realizes the fundamental property of $W$-magic for nonstabilizer HG states. Their entanglement is intrinsically tied to nonstabilizerness and cannot be generated from a magic-infused product state by Clifford operations alone.

\subsection{Dicke states}

Dicke states, which include $W$ states, combine permutation-symmetric entanglement with nonstabilizer structure. We use them as a second class of examples supporting the fundamental property of $W$-magic. Although we do not provide a general proof, previous Clifford-orbit computations and our numerical searches indicate that no Clifford operation maps nontrivial Dicke states to total product states. Dicke states are defined as~\cite{Dicke_1954, Bartschi_2019}
\begin{equation*}
    \ket{D^n_k} = \frac{1}{\sqrt{\binom{n}{k}}} \sum_{w(x) = k} \ket x,
\end{equation*}
where $w(x)$ is the Hamming weight of the bitstring $x$.

\begin{remark}[Dicke states and the fundamental property of $W$-magic] 
    Let $\ket\psi = \ket{D^n_k}$ be a nontrivial Dicke state, with $0 < k < n$ and $n > 2$, and let $C \in \mathcal{C}_n$ be an arbitrary $n$-qubit Clifford gate. Then, for any total product state $\ket{\phi_1} \otimes \dots \otimes \ket{\phi_n}$, 
    \[
        C\ket\psi \neq \bigotimes_{j = 1}^n \ket{\phi_j}.
    \]
    That is, nontrivial Dicke states are expected to exhibit the fundamental property of $W$-magic in Definition~\ref{def:fund_prop_W}.
\end{remark}

Evidence for this conjecture comes from previous Clifford-orbit computations and our numerical searches. Ref.~\cite{Munizzi_2023} found no product representatives in the explored orbits of several small Dicke states, including $\ket{D^3_1} = \ket{W_3}$, $\ket{D^4_2}$, $\ket{D^4_1}$, and $\ket{D^5_1}$. Our searches likewise found no total product representative for Dicke states up to $10$ qubits. Although not an analytical proof, this supports interpreting Dicke states as $W$-magic examples with Clifford-irreducible entanglement.

\paragraph*{$T$-magic contrast.} The preceding examples concern $W$-magic states, whose entanglement is protected against Clifford disentangling. By contrast, the $T$-magic side contains states whose entanglement can be Clifford-removable. To test this, we implemented a simple Clifford-orbit search minimizing the sum of single-qubit entanglement entropies, which vanishes for pure states if and only if the state is a total product state. This heuristic found total product representatives in the Clifford orbits of CTC states up to $10$ qubits; for instance, at $n=8$ it found product representatives for $17$ out of $500$ randomly generated states with nontrivial CTC architectures. This does not characterize the full $T$-magic ensemble, but it confirms that the $T$-magic side contains nonstabilizer states whose entanglement is Clifford-removable, in sharp contrast with the $W$-magic examples above.

\bibliography{references}

\end{document}